\documentclass[submission,copyright,creativecommons]{eptcs}
\providecommand{\event}{QPL 2024} 

\usepackage{iftex}
\usepackage{braket}
\usepackage{tikzit}
\usepackage{amsmath, amsthm}
\usepackage{amsfonts,amssymb}
\usepackage{graphicx,subcaption}
\usepackage{stmaryrd}
\usepackage{array}
\usepackage{algorithm,algpseudocode}

\ifpdf
  \usepackage{underscore}         
  \usepackage[T1]{fontenc}        
\else
  \usepackage{breakurl}           
\fi

\tikzstyle{gate}=[shape=rectangle, text height=1.5ex, text depth=0.25ex, yshift=0.5mm, fill=white, draw=black, minimum height=5mm, yshift=-0.5mm, minimum width=5mm, font={\small}, tikzit category=circuit]
\tikzstyle{big gate}=[shape=rectangle, text height=1.5ex, text depth=0.25ex, yshift=0.5mm, fill=white, draw=black, minimum height=10mm, yshift=-0.5mm, minimum width=5mm, font={\small}, tikzit category=circuit]
\tikzstyle{Z dot}=[inner sep=0mm, minimum size=2mm, shape=circle, draw=black, fill={rgb,255: red,221; green,255; blue,221}, tikzit category=zx]
\tikzstyle{Z phase dot}=[minimum size=5mm, font={\footnotesize\boldmath}, shape=rectangle, rounded corners=2mm, inner sep=0.2mm, outer sep=-2mm, scale=0.8, tikzit shape=rectangle, draw=black, fill={rgb,255: red,221; green,255; blue,221}, tikzit draw=blue, tikzit category=zx]
\tikzstyle{X dot}=[Z dot, shape=circle, draw=black, fill={rgb,255: red,255; green,136; blue,136}, tikzit category=zx]
\tikzstyle{X phase dot}=[Z phase dot, tikzit shape=rectangle, tikzit draw=blue, fill={rgb,255: red,255; green,136; blue,136}, font={\footnotesize\boldmath}, tikzit category=zx]
\tikzstyle{hadamard}=[fill=yellow, draw=black, shape=rectangle, inner sep=0.6mm, minimum height=1.5mm, minimum width=1.5mm, tikzit category=zx]
\tikzstyle{paulibox}=[fill={rgb,255: red,221; green,221; blue,255}, draw=black, shape=rectangle, inner sep=0.6mm, minimum height=5mm, minimum width=5mm, font={\footnotesize}, text height=1.5ex, text depth=0.25ex, tikzit category=zx]
\tikzstyle{vertex}=[inner sep=0mm, minimum size=1mm, shape=circle, draw=black, fill=black, tikzit category=misc]
\tikzstyle{vertex set}=[inner sep=0mm, minimum size=1mm, shape=circle, draw=black, fill=white, font={\footnotesize\boldmath}, tikzit category=misc]
\tikzstyle{small black dot}=[fill=black, draw=black, shape=circle, inner sep=0pt, minimum width=1.2mm, tikzit category=circuit]
\tikzstyle{cnot ctrl}=[fill=black, draw=black, shape=circle, inner sep=0pt, minimum width=1.2mm, tikzit category=circuit]
\tikzstyle{cnot targ}=[fill=white, draw=white, shape=circle, tikzit category=circuit, label={center:$\oplus$}, inner sep=0pt, minimum width=2.1mm, tikzit fill={rgb,255: red,102; green,204; blue,255}, tikzit draw=black]
\tikzstyle{ket}=[fill=white, draw=black, shape=regular polygon, regular polygon sides=3, regular polygon rotate=-30, scale=0.7, inner sep=1pt, tikzit category=circuit, tikzit shape=rectangle, tikzit fill=green]
\tikzstyle{bra}=[fill=white, draw=black, shape=regular polygon, regular polygon sides=3, regular polygon rotate=30, scale=0.7, inner sep=1pt, tikzit category=circuit, tikzit shape=rectangle, tikzit fill=red]
\tikzstyle{scalar}=[shape=rectangle, text height=1.5ex, text depth=0.25ex, yshift=0.5mm, fill=white, draw=black, minimum height=5mm, yshift=-0.5mm, minimum width=5mm, font={\small}, rounded corners=2mm]
\tikzstyle{clabel}=[fill=white, draw=none, shape=rectangle, tikzit fill={rgb,255: red,56; green,255; blue,242}, font={\footnotesize}, inner sep=1pt, tikzit category=labels]
\tikzstyle{empty diagram}=[draw={gray!40!white}, dashed, shape=rectangle, minimum width=1cm, minimum height=1cm, tikzit category=misc]
\tikzstyle{amap}=[fill=white, draw=black, shape=NEbox, tikzit category=asymmetric, tikzit fill=yellow, tikzit shape=rectangle]
\tikzstyle{amap conj}=[fill=white, draw=black, shape=NWbox, tikzit category=asymmetric, tikzit fill=green, tikzit shape=rectangle]
\tikzstyle{amap adj}=[fill=white, draw=black, shape=SEbox, tikzit category=asymmetric, tikzit fill=red, tikzit shape=rectangle]
\tikzstyle{amap trans}=[fill=white, draw=black, shape=SWbox, tikzit category=asymmetric, tikzit fill=orange, tikzit shape=rectangle]
\tikzstyle{astate}=[fill=white, draw=black, shape=NEtriangle, tikzit category=asymmetric, tikzit shape=circle, tikzit fill=yellow]
\tikzstyle{astate conj}=[fill=white, draw=black, shape=NWtriangle, tikzit category=asymmetric, tikzit shape=circle, tikzit fill=green]
\tikzstyle{astate adj}=[fill=white, draw=black, shape=SEtriangle, tikzit category=asymmetric, tikzit shape=circle, tikzit fill=red]
\tikzstyle{astate trans}=[fill=white, draw=black, shape=SWtriangle, tikzit category=asymmetric, tikzit shape=circle, tikzit fill=orange]
\tikzstyle{bigbox}=[fill=white, draw=black, shape=rectangle, tikzit category=zx, minimum width=1.5cm, minimum height=2.6cm]

\tikzstyle{hadamard edge}=[-, dashed, dash pattern=on 2pt off 0.5pt, thick, draw={rgb,255: red,68; green,136; blue,255}]
\tikzstyle{box edge}=[-, dashed, dash pattern=on 2pt off 0.5pt, thick, draw={rgb,255: red,203; green,192; blue,225}]
\tikzstyle{brace edge}=[-, tikzit draw=blue, decorate, decoration={brace,amplitude=1mm,raise=-1mm}]
\tikzstyle{diredge}=[->]
\tikzstyle{double edge}=[-, double, shorten <=-1mm, shorten >=-1mm, double distance=2pt]
\tikzstyle{gray edge}=[-, {gray!60!white}]
\tikzstyle{pointer edge}=[->, very thick, gray]
\tikzstyle{boldedge}=[-, line width=1.6pt, shorten <=-0.17mm, shorten >=-0.17mm]
\tikzstyle{bidir edge}=[<->, very thick, draw={rgb,255: red,191; green,191; blue,191}]
\tikzstyle{separator edge}=[-, dashed, dash pattern=on 2pt off 0.5pt, thick, draw={rgb,255: red,153; green,153; blue,153}]
\tikzstyle{dashed edge}=[-, dashed, dash pattern=on 2pt off 0.5pt, thick]
\tikzstyle{solid edge}=[-, thick]

\newcommand{\image}{\mathrm{Image}}
\newtheorem{lemma}{Lemma}

\title{Fast Classical Simulation of Quantum Circuits via Parametric Rewriting in the ZX-Calculus}
\author{Matthew Sutcliffe
\institute{Department of Computer Science\\ University of Oxford\\ Oxford, UK}
\email{matthew.sutcliffe@cs.ox.ac.uk}
\and
Aleks Kissinger
\institute{Department of Computer Science\\ University of Oxford\\ Oxford, UK}
\email{aleks.kissinger@cs.ox.ac.uk}
}
\def\titlerunning{Fast Classical Simulation via Parametric Rewriting in the ZX-Calculus}
\def\authorrunning{M. Sutcliffe \& A. Kissinger}
\begin{document}
\maketitle

\begin{abstract}
The ZX-calculus is an algebraic formalism that allows quantum computations to be simplified via a small number of simple graphical rewrite rules. Recently, it was shown that, when combined with a family of ``sum-over-Cliffords'' techniques, the ZX-calculus provides a powerful tool for classical simulation of quantum circuits. However, for several important classical simulation tasks, such as computing the probabilities associated with many measurement outcomes of a single quantum circuit, this technique results in reductions over many very similar diagrams, where much of the same computational work is repeated. In this paper, we show that the majority of this work can be shared across branches, by developing reduction strategies that can be run parametrically on diagrams with boolean free parameters. As parameters only need to be fixed after the bulk of the simplification work is already done, we show that it is possible to perform the final stage of classical simulation quickly utilising a high degree of GPU parallelism. Using these methods, we demonstrate an average speedup factor of $78.3\pm10.2$ for certain classical simulation tasks vs. the non-parametric approach.
\end{abstract}

\section{Introduction}

The ZX-calculus~\cite{CD2,DBLP:journals/quant-ph/Wetering20} is a useful tool for expressing a broad range of quantum computations, including quantum circuits, as a type of labelled graph called a \textit{ZX-diagram}, and subsequently reducing it to a simpler form using a handful of graph rewriting rules. Recently, it has been applied extensively in optimisation~\cite{duncan2019graph,Cowtan2020phasegadget,deBeaudrapN2020treducspidernest,borgna2021hybrid,gogioso2023annealing,mcelvanney2023flowpreserving,nagele2023optimizing} and classical simulation~\cite{DBLP:journals/quant-ph/Kissinger21,kissinger2022classical,10.1145/3489517.3530627,Codsi2022Masters,codsi2022classically,cam2023speeding,koch2023speedy,sutcliffe-procopt,koziell-pipe2024towards,ahmad-sutcliffe,ahmad2024,sutcliffeReview} of quantum circuits.

The latter application makes use of the ZX-calculus to simplify circuits as much as possible, before relying on the \textit{sum-over-Cliffords} method of \textit{stabiliser decomposition} introduced by Bravyi, Smith, and Smolin \cite{DBLP:journals/physical-review/BSS} as well as Bravyi et al \cite{bravyi19}. This approach amounts to decomposing a non-Clifford quantum circuit with measurements into an exponentially large sum of efficiently reducible Clifford terms, in order to deduce the probability amplitude of particular measurement outcomes. Applying ZX-calculus simplification at each step of the decomposition, as proposed by Kissinger and van de Wetering \cite{DBLP:journals/quant-ph/Kissinger21}, makes this process more efficient.

In this paper, we identify that when multiple measurement amplitudes or probabilities are needed for the same circuit (such as when sampling the output distribution many times or computing marginal probabilities), the entire procedure needs to be repeated from the start to compute each amplitude independently. We address this issue by adapting the rewriting rules and simplification procedure of the ZX-calculus to support boolean free parameters.

Remarkably, the graph-theoretic simplification of such diagrams can be done in a way that is agnostic to the value of these parameters, which enables us to delay fixing them until near the end of the calculation. Ultimately, this allows –-- once the circuit has been reduced for the generalised case –-- for the respective results of different choices of free parameters to be calculated very rapidly, by simply evaluating a real-valued polynomial with boolean free variables. Furthermore, we show that such polynomials can be represented in such a way that these evaluations can be calculated in parallel on a GPU. Consequently, as the following sections will show, after the initial generalised reduction, every subsequent choice of parameters can be evaluated in a small fraction of the time that would be achieved
by repeating the whole process.

\section{Background}

\subsection{ZX-calculus}
\label{subsec:zx-calc}
Quantum algorithms are commonly expressed graphically in the form of quantum circuit diagrams \cite{DBLP:books/nielsen}, which show the sequence of gates acting upon the qubits. A useful alternative notation is that of the ZX-calculus, wherein quantum computations are expressed as \textit{ZX-diagrams} \cite{DBLP:books/pqp}. These are undirected labelled graphs composed of Z- and X- `\textit{spiders}'. Semantically, Z-spiders can be interpreted ``generalised Z-phase gates'', which are $2^n \times 2^m$ matrices with a 1 in the top-left corner, $e^{i\alpha}$ in the bottom-right, and 0s elsewhere:
\[
\hfill
\left\llbracket \ \begin{tikzpicture}
	\begin{pgfonlayer}{nodelayer}
		\node [style=Z phase dot] (0) at (0, 0) {$\alpha$};
		\node [style=none] (1) at (-1, 0.75) {};
		\node [style=none] (2) at (-1, -0.75) {};
		\node [style=none] (3) at (1, 0.75) {};
		\node [style=none] (4) at (1, -0.75) {};
		\node [style=none, rotate=90] (5) at (-1, 0) {...};
		\node [style=none, rotate=90] (6) at (0.975, 0) {...};
	\end{pgfonlayer}
	\begin{pgfonlayer}{edgelayer}
		\draw [bend left] (1.center) to (0);
		\draw [bend left] (0) to (3.center);
		\draw [bend right] (2.center) to (0);
		\draw [bend right] (0) to (4.center);
	\end{pgfonlayer}
\end{tikzpicture}\  \right\rrbracket
\;:=\;
\begin{pmatrix}
1 & 0 & \cdots & 0 \\
0 & 0 & \cdots & 0 \\
\vdots & \vdots   & \ddots & \\
0  & 0  & & e^{i\alpha} \\
  \end{pmatrix}
  \hfill
\]
Similarly, X-spiders are generalised X-phase gates. One way to define them is by conjugating Z-spiders by Hadamard gates:
\[
\hfill
\begin{tikzpicture}
	\begin{pgfonlayer}{nodelayer}
		\node [style=X phase dot] (0) at (-2.25, 0) {$\alpha$};
		\node [style=none] (1) at (-3.25, 0.75) {};
		\node [style=none] (2) at (-3.25, -0.75) {};
		\node [style=none] (3) at (-1.25, 0.75) {};
		\node [style=none] (4) at (-1.25, -0.75) {};
		\node [style=none, rotate=90] (5) at (-3.25, 0) {...};
		\node [style=none, rotate=90] (6) at (-1.275, 0) {...};
		\node [style=Z phase dot] (7) at (3, 0) {$\alpha$};
		\node [style=none] (8) at (2, 0.75) {};
		\node [style=none] (9) at (2, -0.75) {};
		\node [style=none] (10) at (4, 0.75) {};
		\node [style=none] (11) at (4, -0.75) {};
		\node [style=none, rotate=90] (12) at (1.75, 0) {...};
		\node [style=none, rotate=90] (13) at (4.225, 0) {...};
		\node [style=none] (14) at (0, 0) {$:=$};
		\node [style=hadamard] (15) at (4.25, 0.75) {};
		\node [style=hadamard] (16) at (4.25, -0.75) {};
		\node [style=hadamard] (17) at (1.75, 0.75) {};
		\node [style=hadamard] (18) at (1.75, -0.75) {};
		\node [style=none] (19) at (1.25, 0.75) {};
		\node [style=none] (20) at (1.25, -0.75) {};
		\node [style=none] (21) at (4.75, 0.75) {};
		\node [style=none] (22) at (4.75, -0.75) {};
	\end{pgfonlayer}
	\begin{pgfonlayer}{edgelayer}
		\draw [bend left] (1.center) to (0);
		\draw [bend left] (0) to (3.center);
		\draw [bend right] (2.center) to (0);
		\draw [bend right] (0) to (4.center);
		\draw [in=120, out=0] (8.center) to (7);
		\draw [in=180, out=60] (7) to (10.center);
		\draw [in=240, out=0] (9.center) to (7);
		\draw [in=180, out=-60] (7) to (11.center);
		\draw (19.center) to (8.center);
		\draw (20.center) to (9.center);
		\draw (10.center) to (21.center);
		\draw (11.center) to (22.center);
	\end{pgfonlayer}
\end{tikzpicture}
\qquad\textrm{where}\qquad
\left\llbracket \ \begin{tikzpicture}
	\begin{pgfonlayer}{nodelayer}
		\node [style=none] (0) at (-0.75, 0) {};
		\node [style=none] (1) at (0.75, 0) {};
		\node [style=hadamard] (2) at (0, 0) {};
	\end{pgfonlayer}
	\begin{pgfonlayer}{edgelayer}
		\draw (0.center) to (1.center);
	\end{pgfonlayer}
\end{tikzpicture}
 \ \right\rrbracket
=
\frac{1}{\sqrt{2}}
\begin{pmatrix}
1 & 1 \\
1 & -1
\end{pmatrix}
\hfill
\]
More complex ZX-diagrams can be interpreted much like quantum circuits in terms of composition and tensor products of their components:
\[
\hfill
\left\llbracket \ \begin{tikzpicture}
	\begin{pgfonlayer}{nodelayer}
		\node [style=big gate] (0) at (0.75, 0) {$E$};
		\node [style=none] (2) at (1.5, 0.75) {};
		\node [style=none] (4) at (1.5, -0.75) {};
		\node [style=none, rotate=90] (6) at (1.5, 0) {...};
		\node [style=big gate] (7) at (-0.75, 0) {$D$};
		\node [style=none] (8) at (-1.5, 0.75) {};
		\node [style=none] (9) at (0, 0.75) {};
		\node [style=none] (10) at (-1.5, -0.75) {};
		\node [style=none] (11) at (0, -0.75) {};
		\node [style=none, rotate=90] (12) at (-1.5, 0) {...};
		\node [style=none, rotate=90] (13) at (0, 0) {...};
	\end{pgfonlayer}
	\begin{pgfonlayer}{edgelayer}
		\draw (8.center) to (9.center);
		\draw (10.center) to (11.center);
		\draw (9.center) to (2.center);
		\draw (11.center) to (4.center);
	\end{pgfonlayer}
\end{tikzpicture} \ \right\rrbracket
=
\left\llbracket \ \begin{tikzpicture}
	\begin{pgfonlayer}{nodelayer}
		\node [style=big gate] (0) at (0, 0) {$E$};
		\node [style=none] (1) at (-0.75, 0.75) {};
		\node [style=none] (2) at (0.75, 0.75) {};
		\node [style=none] (3) at (-0.75, -0.75) {};
		\node [style=none] (4) at (0.75, -0.75) {};
		\node [style=none, rotate=90] (5) at (-0.75, 0) {...};
		\node [style=none, rotate=90] (6) at (0.75, 0) {...};
	\end{pgfonlayer}
	\begin{pgfonlayer}{edgelayer}
		\draw (1.center) to (2.center);
		\draw (3.center) to (4.center);
	\end{pgfonlayer}
\end{tikzpicture} \ \right\rrbracket \circ
\left\llbracket \ \begin{tikzpicture}
	\begin{pgfonlayer}{nodelayer}
		\node [style=big gate] (0) at (0, 0) {$D$};
		\node [style=none] (1) at (-0.75, 0.75) {};
		\node [style=none] (2) at (0.75, 0.75) {};
		\node [style=none] (3) at (-0.75, -0.75) {};
		\node [style=none] (4) at (0.75, -0.75) {};
		\node [style=none, rotate=90] (5) at (-0.75, 0) {...};
		\node [style=none, rotate=90] (6) at (0.75, 0) {...};
	\end{pgfonlayer}
	\begin{pgfonlayer}{edgelayer}
		\draw (1.center) to (2.center);
		\draw (3.center) to (4.center);
	\end{pgfonlayer}
\end{tikzpicture} \ \right\rrbracket
\quad
\left\llbracket \ \begin{tikzpicture}
	\begin{pgfonlayer}{nodelayer}
		\node [style=big gate] (0) at (0, -1.125) {$E$};
		\node [style=none] (1) at (-0.75, -0.375) {};
		\node [style=none] (2) at (0.75, -0.375) {};
		\node [style=none] (3) at (-0.75, -1.875) {};
		\node [style=none] (4) at (0.75, -1.875) {};
		\node [style=none, rotate=90] (5) at (-0.75, -1.125) {...};
		\node [style=none, rotate=90] (6) at (0.75, -1.125) {...};
		\node [style=big gate] (7) at (0, 1.125) {$D$};
		\node [style=none] (8) at (-0.75, 1.875) {};
		\node [style=none] (9) at (0.75, 1.875) {};
		\node [style=none] (10) at (-0.75, 0.375) {};
		\node [style=none] (11) at (0.75, 0.375) {};
		\node [style=none, rotate=90] (12) at (-0.75, 1.125) {...};
		\node [style=none, rotate=90] (13) at (0.75, 1.125) {...};
	\end{pgfonlayer}
	\begin{pgfonlayer}{edgelayer}
		\draw (1.center) to (2.center);
		\draw (3.center) to (4.center);
		\draw (8.center) to (9.center);
		\draw (10.center) to (11.center);
	\end{pgfonlayer}
\end{tikzpicture} \ \right\rrbracket
=
\left\llbracket \ \begin{tikzpicture}
	\begin{pgfonlayer}{nodelayer}
		\node [style=big gate] (0) at (0, 0) {$D$};
		\node [style=none] (1) at (-0.75, 0.75) {};
		\node [style=none] (2) at (0.75, 0.75) {};
		\node [style=none] (3) at (-0.75, -0.75) {};
		\node [style=none] (4) at (0.75, -0.75) {};
		\node [style=none, rotate=90] (5) at (-0.75, 0) {...};
		\node [style=none, rotate=90] (6) at (0.75, 0) {...};
	\end{pgfonlayer}
	\begin{pgfonlayer}{edgelayer}
		\draw (1.center) to (2.center);
		\draw (3.center) to (4.center);
	\end{pgfonlayer}
\end{tikzpicture} \ \right\rrbracket \otimes
\left\llbracket \ \begin{tikzpicture}
	\begin{pgfonlayer}{nodelayer}
		\node [style=big gate] (0) at (0, 0) {$E$};
		\node [style=none] (1) at (-0.75, 0.75) {};
		\node [style=none] (2) at (0.75, 0.75) {};
		\node [style=none] (3) at (-0.75, -0.75) {};
		\node [style=none] (4) at (0.75, -0.75) {};
		\node [style=none, rotate=90] (5) at (-0.75, 0) {...};
		\node [style=none, rotate=90] (6) at (0.75, 0) {...};
	\end{pgfonlayer}
	\begin{pgfonlayer}{edgelayer}
		\draw (1.center) to (2.center);
		\draw (3.center) to (4.center);
	\end{pgfonlayer}
\end{tikzpicture} \ \right\rrbracket
\hfill \]
and crossing wires are interpreted as ``swap'' gates: $\textrm{SWAP}(v \otimes w) = w \otimes v$.
Whereas quantum circuits always describe unitary matrices, ZX-diagrams can in general be non-unitary.
In particular, a ZX-diagram with $m$ inputs and $n$ outputs describes a $2^n \times 2^m$ matrix. If it has no inputs or outputs, it describes a $2^0 \times 2^0 = 1 \times 1$ matrix, i.e. a scalar.
An important property of ZX-diagrams is that isomorphic graphs semantically describe the same linear map. This is sometimes stated as a ``meta-rule'' of ZX-calculus: \textit{only connectivity matters}.


The ZX-calculus comes equipped with a number of \textit{rewrite rules} that describe how certain structures within a ZX-diagram may be re-expressed as an equivalent (generally simplified) structure. Figure \ref{fig:basicrules} shows the elementary rewrite rules. The rules in Figure \ref{fig:basicrules} are sufficient to efficiently reduce any \textit{Clifford} ZX-diagrams (those diagrams whose angles are restricted to integer multiples of $\frac{\pi}{2}$~\cite{Backens1}) with no open inputs or output to a scalar.

\begin{figure}[h]
\begin{tikzpicture}
	\begin{pgfonlayer}{nodelayer}
		\node [style=none] (0) at (-2, 1) {};
		\node [style=none] (1) at (-2, -0.25) {};
		\node [style=Z phase dot] (4) at (-0.75, 0.25) {$\alpha$};
		\node [style=none] (6) at (0.5, 1) {};
		\node [style=none] (8) at (0.25, 0.5) {$\vdots$};
		\node [style=none] (9) at (-1.75, 0.5) {$\vdots$};
		\node [style=none] (10) at (1.5, -0.75) {$=$};
		\node [style=none] (12) at (-2, 1.5) {};
		\node [style=none] (13) at (0.5, 1.5) {};
		\node [style=none] (14) at (0.5, -0.25) {};
		\node [style=none] (15) at (-2, -2.5) {};
		\node [style=none] (16) at (-2, -3) {};
		\node [style=Z phase dot] (17) at (-0.75, -1.75) {$\beta$};
		\node [style=none] (18) at (0.5, -2.5) {};
		\node [style=none] (19) at (0.25, -1.75) {$\vdots$};
		\node [style=none] (20) at (-1.75, -1.75) {$\vdots$};
		\node [style=none] (21) at (-2, -1.25) {};
		\node [style=none] (22) at (0.5, -1.25) {};
		\node [style=none] (23) at (0.5, -3) {};
		\node [style=none] (24) at (-0.75, -0.75) {$...$};
		\node [style=none] (25) at (2.5, 0.75) {};
		\node [style=Z phase dot] (27) at (3.75, -0.75) {$\;\alpha+\beta\;$};
		\node [style=none] (28) at (5, 0.75) {};
		\node [style=none] (31) at (2.5, 1.5) {};
		\node [style=none] (32) at (5, 1.5) {};
		\node [style=none] (38) at (5, -0.5) {$\vdots$};
		\node [style=none] (39) at (2.5, -0.5) {$\vdots$};
		\node [style=none] (40) at (2.5, -3) {};
		\node [style=none] (41) at (5, -3) {};
		\node [style=none] (42) at (2.5, -2.25) {};
		\node [style=none] (43) at (5, -2.25) {};
		\node [style=Z phase dot] (46) at (-0.75, -5.75) {$\alpha$};
		\node [style=none] (47) at (0.5, -4.75) {};
		\node [style=none] (48) at (0.25, -5.75) {$\vdots$};
		\node [style=none] (50) at (-2, -5.75) {};
		\node [style=none] (51) at (0.5, -5.25) {};
		\node [style=none] (52) at (0.5, -6.75) {};
		\node [style=hadamard] (53) at (-1.5, -5.75) {};
		\node [style=none] (57) at (1.5, -5.75) {$=$};
		\node [style=X phase dot] (58) at (3.25, -5.75) {$\alpha$};
		\node [style=hadamard] (59) at (4.5, -4.75) {};
		\node [style=none] (60) at (4.25, -5.75) {$\vdots$};
		\node [style=hadamard] (62) at (4.5, -5.25) {};
		\node [style=hadamard] (63) at (4.5, -6.75) {};
		\node [style=none] (64) at (2.5, -5.75) {};
		\node [style=none] (65) at (5, -4.75) {};
		\node [style=none] (66) at (5, -5.25) {};
		\node [style=none] (67) at (5, -6.75) {};
		\node [style=none] (78) at (10.75, -0.75) {$=$};
		\node [style=Z phase dot] (79) at (13.5, -0.75) {$\;\alpha-2a\alpha\;$};
		\node [style=X phase dot] (80) at (15, 1) {$a\pi$};
		\node [style=none] (81) at (15, -1) {$\vdots$};
		\node [style=none] (82) at (11.75, -0.75) {};
		\node [style=X phase dot] (83) at (15, 0) {$a\pi$};
		\node [style=X phase dot] (84) at (15, -2.25) {$a\pi$};
		\node [style=none] (85) at (12.75, -0.75) {};
		\node [style=none] (86) at (15.75, 1) {};
		\node [style=none] (87) at (15.75, 0) {};
		\node [style=none] (88) at (15.75, -2.25) {};
		\node [style=none] (89) at (12.75, -1.75) {$e^{ia\alpha}$};
		\node [style=Z phase dot] (90) at (8.5, -0.75) {$\alpha$};
		\node [style=none] (91) at (9.5, 1) {};
		\node [style=none] (92) at (9.5, -1) {$\vdots$};
		\node [style=none] (93) at (6.75, -0.75) {};
		\node [style=none] (94) at (9.5, 0) {};
		\node [style=none] (95) at (9.5, -2.25) {};
		\node [style=X phase dot] (96) at (7.5, -0.75) {$a\pi$};
		\node [style=none] (100) at (20.5, -0.75) {$=$};
		\node [style=X phase dot] (102) at (24, 1) {$a\pi$};
		\node [style=none] (103) at (24, -1) {$\vdots$};
		\node [style=X phase dot] (105) at (24, 0) {$a\pi$};
		\node [style=X phase dot] (106) at (24, -2.25) {$a\pi$};
		\node [style=none] (108) at (25, 1) {};
		\node [style=none] (109) at (25, 0) {};
		\node [style=none] (110) at (25, -2.25) {};
		\node [style=none] (111) at (22.5, -0.75) {$\frac{e^{ia\alpha}}{\sqrt{2}^{(n-1)}}$};
		\node [style=Z phase dot] (112) at (18.75, -0.75) {$\alpha$};
		\node [style=none] (113) at (19.75, 1) {};
		\node [style=none] (114) at (19.75, -1) {$\vdots$};
		\node [style=none] (116) at (19.75, 0) {};
		\node [style=none] (117) at (19.75, -2.25) {};
		\node [style=X phase dot] (118) at (17.75, -0.75) {$a\pi$};
		\node [style=none] (122) at (7, -5) {};
		\node [style=none] (123) at (7, -6.5) {};
		\node [style=X dot] (124) at (7.75, -5.75) {};
		\node [style=Z dot] (125) at (9, -5.75) {};
		\node [style=none] (126) at (9.75, -5) {};
		\node [style=none] (127) at (9.75, -6.5) {};
		\node [style=none] (128) at (10.75, -5.75) {$=$};
		\node [style=Z dot] (129) at (13.75, -5) {};
		\node [style=X dot] (130) at (15, -5) {};
		\node [style=none] (131) at (13, -5) {};
		\node [style=none] (132) at (15.75, -5) {};
		\node [style=Z dot] (133) at (13.75, -6.5) {};
		\node [style=X dot] (134) at (15, -6.5) {};
		\node [style=none] (135) at (13, -6.5) {};
		\node [style=none] (136) at (15.75, -6.5) {};
		\node [style=none] (137) at (12, -5.75) {$\sqrt{2}$};
		\node [style=none] (138) at (18, -3.75) {};
		\node [style=none] (139) at (20.5, -3.75) {};
		\node [style=Z dot] (140) at (19.25, -3.75) {};
		\node [style=none] (141) at (21.5, -3.75) {$=$};
		\node [style=none] (144) at (18, -5.25) {};
		\node [style=none] (145) at (20.5, -5.25) {};
		\node [style=none] (147) at (21.5, -5.25) {$=$};
		\node [style=hadamard] (150) at (18.75, -5.25) {};
		\node [style=hadamard] (151) at (19.75, -5.25) {};
		\node [style=none] (152) at (22.5, -3.75) {};
		\node [style=none] (153) at (25, -3.75) {};
		\node [style=none] (155) at (22.5, -5.25) {};
		\node [style=none] (156) at (25, -5.25) {};
		\node [style=Z phase dot] (157) at (-1.5, -8.5) {$\alpha$};
		\node [style=X phase dot] (158) at (0, -8.5) {$\beta$};
		\node [style=none] (159) at (1.5, -8.5) {$=$};
		\node [style=none] (160) at (6.75, -8.5) {$\frac{1}{\sqrt{2}}\left(1+e^{i\alpha}+e^{i\beta}-e^{i(\alpha+\beta)}\right)$};
		\node [style=Z phase dot] (161) at (20.25, -6.75) {$\alpha$};
		\node [style=none] (162) at (21.5, -6.75) {$=$};
		\node [style=none] (163) at (23.25, -6.75) {$1+e^{i\alpha}$};
		\node [style=none] (164) at (12, -8.5) {};
		\node [style=none] (165) at (14, -8.5) {};
		\node [style=hadamard] (166) at (13, -8.5) {};
		\node [style=none] (167) at (14.75, -8.5) {$\equiv$};
		\node [style=none] (168) at (17, -8.5) {};
		\node [style=none] (169) at (21, -8.5) {};
		\node [style=none] (171) at (22, -8.5) {$\equiv$};
		\node [style=none] (172) at (23, -8.5) {};
		\node [style=none] (173) at (25, -8.5) {};
		\node [style=Z phase dot] (174) at (18, -8.5) {$\frac{\pi}{2}$};
		\node [style=Z phase dot] (175) at (20, -8.5) {$\frac{\pi}{2}$};
		\node [style=X phase dot] (176) at (19, -8.5) {$\frac{\pi}{2}$};
		\node [style=none] (177) at (16, -8.5) {$e^{-i\frac{\pi}{4}}$};
	\end{pgfonlayer}
	\begin{pgfonlayer}{edgelayer}
		\draw [bend left] (0.center) to (4);
		\draw [bend left] (4) to (6.center);
		\draw [bend left] (4) to (1.center);
		\draw [bend left, looseness=1.25] (12.center) to (4);
		\draw [bend left, looseness=1.25] (4) to (13.center);
		\draw [bend right=15] (4) to (14.center);
		\draw [bend right] (15.center) to (17);
		\draw [bend right] (17) to (18.center);
		\draw [bend left, looseness=1.25] (17) to (16.center);
		\draw [bend left] (21.center) to (17);
		\draw [bend left] (17) to (22.center);
		\draw [bend right, looseness=1.25] (17) to (23.center);
		\draw [bend right] (4) to (17);
		\draw [bend left] (4) to (17);
		\draw [in=90, out=0] (25.center) to (27);
		\draw [in=-180, out=90] (27) to (28.center);
		\draw [in=90, out=0] (31.center) to (27);
		\draw [in=-180, out=90] (27) to (32.center);
		\draw [in=0, out=-90] (27) to (42.center);
		\draw [in=270, out=0] (40.center) to (27);
		\draw [in=-180, out=-90] (27) to (43.center);
		\draw [in=270, out=180] (41.center) to (27);
		\draw [bend right=45] (46) to (52.center);
		\draw [bend left=45] (46) to (47.center);
		\draw [in=180, out=45] (46) to (51.center);
		\draw (46) to (53);
		\draw (53) to (50.center);
		\draw [bend right=45] (58) to (63);
		\draw [bend left=45] (58) to (59);
		\draw [in=180, out=45] (58) to (62);
		\draw (58) to (64.center);
		\draw (63) to (67.center);
		\draw (62) to (66.center);
		\draw (59) to (65.center);
		\draw [bend right=45] (79) to (84);
		\draw [bend left=45] (79) to (80);
		\draw [in=180, out=60] (79) to (83);
		\draw (79) to (85.center);
		\draw (85.center) to (82.center);
		\draw (84) to (88.center);
		\draw (83) to (87.center);
		\draw (80) to (86.center);
		\draw [bend right] (90) to (95.center);
		\draw [bend left] (90) to (91.center);
		\draw [in=180, out=60] (90) to (94.center);
		\draw (90) to (96);
		\draw (96) to (93.center);
		\draw (106) to (110.center);
		\draw (105) to (109.center);
		\draw (102) to (108.center);
		\draw [bend right] (112) to (117.center);
		\draw [bend left] (112) to (113.center);
		\draw [in=180, out=60] (112) to (116.center);
		\draw (112) to (118);
		\draw [bend left] (122.center) to (124);
		\draw [bend left] (124) to (123.center);
		\draw (124) to (125);
		\draw [bend left] (125) to (126.center);
		\draw [bend right] (125) to (127.center);
		\draw (131.center) to (129);
		\draw (129) to (130);
		\draw (130) to (132.center);
		\draw (135.center) to (133);
		\draw (133) to (134);
		\draw (134) to (136.center);
		\draw (129) to (134);
		\draw (130) to (133);
		\draw (138.center) to (140);
		\draw (140) to (139.center);
		\draw (144.center) to (150);
		\draw (150) to (151);
		\draw (151) to (145.center);
		\draw (155.center) to (156.center);
		\draw (152.center) to (153.center);
		\draw (157) to (158);
		\draw (164.center) to (166);
		\draw (166) to (165.center);
		\draw [style=hadamard edge] (172.center) to (173.center);
		\draw (168.center) to (174);
		\draw (174) to (176);
		\draw (176) to (175);
		\draw (175) to (169.center);
	\end{pgfonlayer}
\end{tikzpicture}
\centering
\caption{A set of the basic rewrite rules and scalar relations of ZX-calculus, where Greek letters denote arbitrary real variables, $[0,2\pi)$, and Latin letters arbitrary boolean variables, $\{0,1\}$. Note that the rules still apply if all the spider colours are inverted.}
\label{fig:basicrules}
\end{figure}
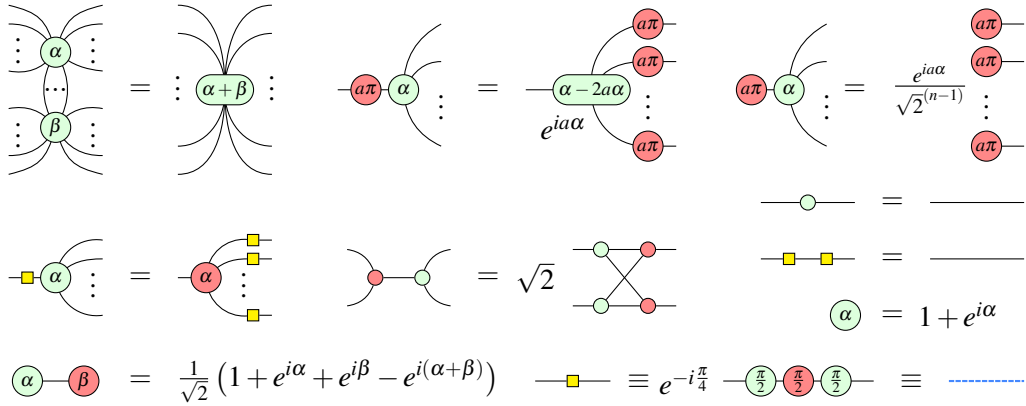

When extending to the Clifford+T gateset, which is sufficient to approximately express any quantum map \cite{Backens1,ng-completeness-2018,BDLP:theses/wang}, this efficiency is lost. When simplifying ZX-diagrams containing T-like gates (corresponding to spiders with phases of \textit{odd} multiples of $\frac{\pi}{4}$), some of these T-spiders may be liable to fuse and cancel, or be otherwise transferred from the diagram to the scalar factor, but many will typically hinder further simplification. When reaching such apparent dead-ends, however, one may continue the diagram's reduction by decomposing the T-gates into sums of Clifford terms.

The BSS decomposition, for instance (introduced in \cite{DBLP:journals/physical-review/BSS}) allows sets of $6$ T-spiders to be exchanged for sums of $7$ Cliffords. As per \cite{DBLP:journals/quant-ph/Kissinger21}, this can be expressed in ZX-calculus form as in Figure \ref{fig:bssdecomp}. Consequently, utilising such a decomposition to remove T-gates, and hence allowing the circuit to be further reduced, results in an exponential increase in the number of terms to compute (and correspondingly the runtime) versus the number of T-gates, $t$, in the initial simplified circuit. In the case of the BSS decomposition, this complexity is given by $7^{t/6}\approx2^{0.47t}$. Hence, for a general decomposition complexity, $2^{\alpha t}$, BSS has an $\alpha\approx0.47$ (where lower values of $\alpha$ represent more efficient decompositions). The current state of the art T-decompositions can achieve $\alpha \approx 0.396$~\cite{Qassim2021improvedupperbounds}.

\begin{figure}[h]
\input{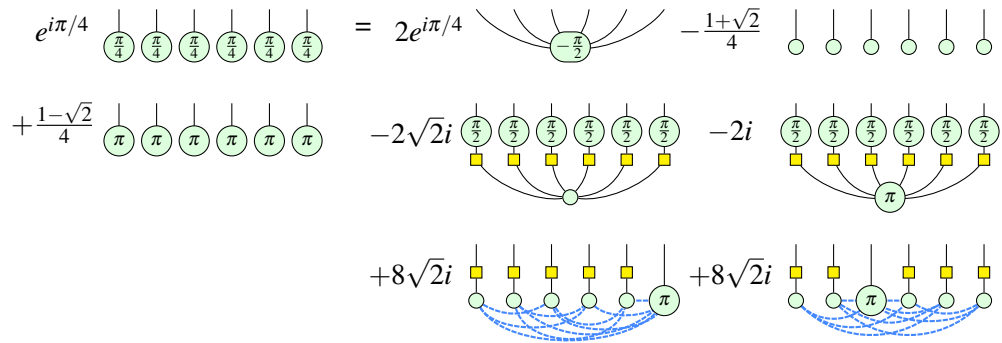}
\centering
\caption{The BSS decomposition \cite{DBLP:journals/physical-review/BSS}, relating a set of $6$ T-gates to a sum of $7$ Clifford terms, expressed in ZX-calculus terms as per \cite{DBLP:journals/quant-ph/Kissinger21}.}
\label{fig:bssdecomp}
\end{figure}

Furthermore, as explored in \cite{DBLP:journals/quant-ph/Kissinger21}, at every stage of applying the decomposition to reduce the T-count (number of T-gates), the rewrite rules of ZX-calculus may be applied in order to, where possible, remove any T-gates that are now liable for removal. Given this extra step, the $\alpha$ values may be seen to represent, fairly na\"ively, an estimate of the \textit{upper-bound} of the number of terms that a Clifford+T circuit will decompose to.

\subsection{Classical simulation}
\label{subsec:classicsim}

A very prominent task in the field is that of classically simulating quantum circuits, in order to, among other things, verify quantum algorithms and hardware. More specifically, this can be divided into \textit{strong} and \textit{weak} simulation, whereby the former is to deduce arbitrary (potentially marginal) probabilities associated with specific measurement outcomes, while the latter is to produce sample outputs of a circuit in accordance with its output probability distribution.

As the name suggests, being able to perform strong simulation implies we can do weak simulation. However, this relies crucially on having access to marginal probabilities. Suppose one wishes to sample a bitstring $(a_1, \ldots, a_n)$ from the probability distribution associated with measuring the output of a quantum circuit. One can first compute the marginal probability $P(x_1 = 0)$, then set the first bit of the output $a_1 := 0$ with probability $P(x_1 = 0)$ and otherwise set $a_1 := 1$. We can compute the marginal probability $P(x_1 = a_1, x_2 = 0)$ and set the next bit $a_2 := 0$ with probability $P(x_1 = a_1, x_2 = 0)/P(x_1 = a_1)$ and set $a_2 := 1$ otherwise, and so on. Using the product rule, we see that the resulting bitstring will be sampled according to the distribution $P(x_1, \ldots, x_n)$.

A circuit may be strongly simulated by plugging a fixed input state into the inputs and the conjugate-transpose of the state associated with a particular measurement outcome into the outputs. The resulting closed diagram can then be reduced to a single complex number $\lambda$ via the means outlined in section \ref{subsec:zx-calc}. The quantity $|\lambda|^2$ then corresponds to the probability associated with that particular measurement outcome computed according to the Born rule of quantum theory. One may also deduce marginal probabilities, concerning measurement outcomes where some qubit outputs ($a_{1},...,a_{k}$) are set, while not caring about the others ($b_{1},...,b_{m}$). Two means of denoting this in ZX-calculus (namely the `\textit{summing}' and `\textit{doubling}' approaches) are shown in Figure \ref{fig:classicsim} \cite{DBLP:journals/quant-ph/Kissinger21}. As to which of these two options is more efficient varies from circuit to circuit, depending on a few factors, including the number of qubits and how much cancellation is facilitated by partially composing against its own adjoint.

Both methods for computing marginal probabilities potentially incur an exponential overhead. The `summing' approach requires summing over $2^m$ distinct terms, whereas the `doubling' approach takes advantage of properties of the Born rule and ZX-diagrams to eliminate the sum, but doubles the size of the diagram, which for a Clifford+T circuit will double its T-count.


\begin{figure}[h]
\begin{tikzpicture}
	\begin{pgfonlayer}{nodelayer}
		\node [style=bigbox] (0) at (0, 0) {$U$};
		\node [style=X phase dot] (1) at (3, 2.25) {$\;a_{1}\pi\;$};
		\node [style=none] (2) at (1, 2.25) {};
		\node [style=X phase dot] (3) at (3, 0.75) {$\;a_{k}\pi\;$};
		\node [style=none] (4) at (1, 0.75) {};
		\node [style=X phase dot] (5) at (3, -0.75) {$\;b_{1}\pi\;$};
		\node [style=none] (6) at (1, -0.75) {};
		\node [style=X phase dot] (7) at (3, -2.25) {$\;b_{m}\pi\;$};
		\node [style=none] (8) at (1, -2.25) {};
		\node [style=X dot] (10) at (-2.75, 2.25) {};
		\node [style=X dot] (12) at (-2.75, 0.75) {};
		\node [style=X dot] (14) at (-2.75, -0.75) {};
		\node [style=X dot] (16) at (-2.75, -2.25) {};
		\node [style=none] (17) at (-1, 2.25) {};
		\node [style=none] (18) at (-1, 0.75) {};
		\node [style=none] (19) at (-1, -0.75) {};
		\node [style=none] (20) at (-1, -2.25) {};
		\node [style=none] (21) at (2, -1.25) {$\vdots$};
		\node [style=none] (22) at (-2, -1.25) {$\vdots$};
		\node [style=bigbox] (23) at (10, 0) {$U$};
		\node [style=X phase dot] (24) at (13, 2.25) {$\;a_{1}\pi\;$};
		\node [style=none] (25) at (11, 2.25) {};
		\node [style=X phase dot] (26) at (13, 0.75) {$\;a_{k}\pi\;$};
		\node [style=none] (27) at (11, 0.75) {};
		\node [style=none] (29) at (11, -0.75) {};
		\node [style=none] (31) at (11, -2.25) {};
		\node [style=X dot] (32) at (7.25, 2.25) {};
		\node [style=X dot] (33) at (7.25, 0.75) {};
		\node [style=X dot] (34) at (7.25, -0.75) {};
		\node [style=X dot] (35) at (7.25, -2.25) {};
		\node [style=none] (36) at (9, 2.25) {};
		\node [style=none] (37) at (9, 0.75) {};
		\node [style=none] (38) at (9, -0.75) {};
		\node [style=none] (39) at (9, -2.25) {};
		\node [style=none] (40) at (12, -1.25) {$\vdots$};
		\node [style=none] (41) at (8, -1.25) {$\vdots$};
		\node [style=bigbox] (42) at (18, 0) {$U^{\dag}$};
		\node [style=X dot] (43) at (20.75, 2.25) {};
		\node [style=none] (44) at (19, 2.25) {};
		\node [style=X dot] (45) at (20.75, 0.75) {};
		\node [style=none] (46) at (19, 0.75) {};
		\node [style=none] (47) at (19, -0.75) {};
		\node [style=none] (48) at (19, -2.25) {};
		\node [style=X phase dot] (49) at (15, 2.25) {$\;a_{1}\pi\;$};
		\node [style=X phase dot] (50) at (15, 0.75) {$\;a_{k}\pi\;$};
		\node [style=none] (53) at (17, 2.25) {};
		\node [style=none] (54) at (17, 0.75) {};
		\node [style=none] (55) at (17, -0.75) {};
		\node [style=none] (56) at (17, -2.25) {};
		\node [style=none] (57) at (20, -1.25) {$\vdots$};
		\node [style=none] (58) at (16, -1.25) {$\vdots$};
		\node [style=none] (59) at (12, 1.75) {$\vdots$};
		\node [style=none] (60) at (8, 1.75) {$\vdots$};
		\node [style=none] (61) at (20, 1.75) {$\vdots$};
		\node [style=none] (62) at (16, 1.75) {$\vdots$};
		\node [style=X dot] (63) at (20.75, -0.75) {};
		\node [style=X dot] (64) at (20.75, -2.25) {};
		\node [style=none] (65) at (-3.5, 3.25) {};
		\node [style=none] (66) at (-3.5, -3.25) {};
		\node [style=none] (67) at (4.25, 3.25) {};
		\node [style=none] (68) at (4.25, -3.25) {};
		\node [style=none] (69) at (4.75, 3.5) {$2$};
		\node [style=none] (70) at (2, 1.75) {$\vdots$};
		\node [style=none] (71) at (-2, 1.75) {$\vdots$};
		\node [style=none] (72) at (5.75, 0) {$\approx$};
		\node [style=none] (73) at (-5, 0) {$\displaystyle\sum\limits_{b_{1},...,b_{m}}$};
	\end{pgfonlayer}
	\begin{pgfonlayer}{edgelayer}
		\draw (1) to (2.center);
		\draw (3) to (4.center);
		\draw (5) to (6.center);
		\draw (7) to (8.center);
		\draw (17.center) to (10);
		\draw (12) to (18.center);
		\draw (19.center) to (14);
		\draw (16) to (20.center);
		\draw (24) to (25.center);
		\draw (26) to (27.center);
		\draw (36.center) to (32);
		\draw (33) to (37.center);
		\draw (38.center) to (34);
		\draw (35) to (39.center);
		\draw (43) to (44.center);
		\draw (45) to (46.center);
		\draw (53.center) to (49);
		\draw (50) to (54.center);
		\draw (64) to (48.center);
		\draw (47.center) to (63);
		\draw (29.center) to (55.center);
		\draw (31.center) to (56.center);
		\draw (65.center) to (66.center);
		\draw (67.center) to (68.center);
	\end{pgfonlayer}
\end{tikzpicture}
\centering
\caption{Two means of denoting as a ZX-diagram (up to a scalar factor) the probability that a circuit, $U$, will produce the qubit outcomes $a_{1},...,a_{k}$, irrespective of the qubit outcomes $b_{1},...,b_{m}$ \cite{DBLP:journals/quant-ph/Kissinger21}. Note that $k+m=n$, for a circuit of $n$ total qubits.}
\label{fig:classicsim}
\end{figure}
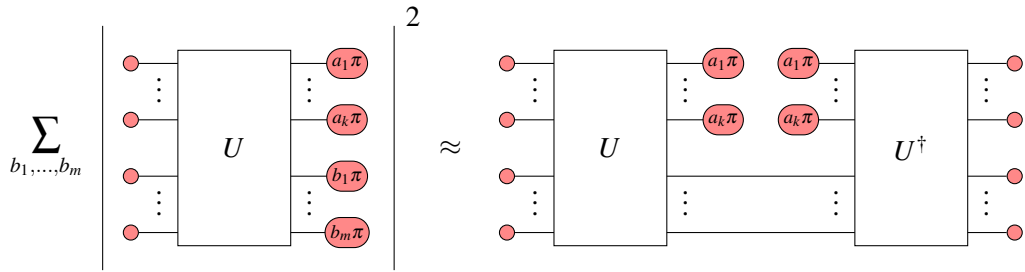


\subsection{GPU parallelism}

A GPU is a specialised hardware component optimised for computing `\textit{Single Instruction Multiple Data}' (\textit{SIMD}) tasks very efficiently. Equipped with many thousands of parallel \textit{threads}, a GPU is adept at performing simple functions, relying on basic arithmetic, to many units of data simultaneously \cite{cuda}. Importantly, as the instructions are executed in lock-step across the threads\footnote{This is a slight oversimplification, as different `\textit{warps}' of $32$ threads may execute out of sync.}, branching code (such as produced by `If' statements) should be avoided, and as each thread is much less powerful than a CPU core, the GPU functions (or `\textit{kernels}') should rely only upon simple arithmetic.

Note that there is a runtime overhead involved in physically transferring the data across hardware (from the CPU to the GPU and back). If the number of data elements to process exceeds the number of available parallel threads, then this data transfer may be `\textit{pipelined}'. This means sending an initial batch of data and beginning its processing while the next batch is in transit. When applicable, this may render the data transfer time negligible and avoid running out of memory space on the GPU.

\section{Methods}

\subsection{Parameterising ZX-calculus}

Computing a $k$-qubit marginal probability of a quantum circuit via the summation method (left-hand side of Figure \ref{fig:classicsim}) involves reducing $2^k$ almost identical ZX-diagrams to scalars and summing the results. For circuits of relatively high T-counts, $t$, each such reduction may be very slow, requiring the computation of up to $2^{\alpha t}$ stabiliser terms, with $\alpha$ given by the decomposition efficiency.

Consequently, devising a means of reducing such similar ZX-diagrams as one, rather than reducing each independently, would be of great utility. Specifically, such sets of ZX-diagrams are structurally identical and differing only insofar as some spiders vary by a phase of $\pm\pi$. We may call such sets of ZX-diagrams `\textit{parametrically symmetric}'. Any parametrically symmetric set may therefore be expressed as a single parameterised ZX-diagram, with every spider phase restricted to $\alpha\in\mathbb{R}$ or:

\begin{equation}
    \label{eqn:param-phase}
    \Phi = (\mathfrak{p}_1 \oplus \mathfrak{p}_2 \oplus\ldots\oplus \mathfrak{p}_n)\pi+\alpha
\end{equation}

where $\oplus$ denotes the XOR operator (addition modulo $2$), $\alpha\in\mathbb{R}$, and $\mathfrak{p}_i\;\forall i=1,2,\ldots,n$ for $n\geq1$ are uninstantiated boolean parameters, such that $\mathfrak{p}_i$ is unfixed but may later be instantiated to $\mathfrak{p}_i\rightarrow p_i$, where $p_i\in\mathbb{B}$. As such, while the parameter takes no fixed value, its \textit{image} (set of possible values to which it may later be instantiated) is known: $\image(\mathfrak{p}_i)=\{0,1\}$. Given this formalism, any parameterised phase, which necessarily takes the form of Equation \ref{eqn:param-phase}, is such that:

\begin{equation}
    \image(\Phi)=\{\alpha,\alpha+\pi\}
\end{equation}
given a particular $\alpha\in\mathbb{R}$. Every phase, therefore, is either fixed and known, $\alpha\in\mathbb{R}$, or parameterised such that it may later be instantiated to $\alpha$ or $\alpha+\pi$.\footnote{Appendix \ref{app:no-ambig-rule} provides further justification for this restriction.} Given this restriction, the parameterised phases may be said to be `\textit{polarised}' and such ZX-diagrams said to be `\textit{polar-parameterised}'.

Hence, any parametrically symmetric set of ZX-diagrams, expressed as a single polar-parameterised ZX-diagram, may be simplified as a single entity via the polar-parameterised versions of the rewriting rules summarised in Figure \ref{fig:zx-rules-param}. Each is derivable from their non-parametric counterparts. Polar-parameterisations of additional rewriting rules, derivable from this fundamental set, are also included in Appendix \ref{app:more-rules}. Lastly, note that every phase is implicitly modulo $2\pi$ and $\image(\Phi)\subseteq\{0,\pi\}$, for instance, means the phase $\Phi$ may either be fixed at $0$ or $\pi$ or indeed polar-parameterised as $\Phi=(\mathfrak{p}_1\oplus\mathfrak{p}_2\oplus\ldots\oplus\mathfrak{p}_n)\pi$ such that $\image(\Phi)=\{0,\pi\}$.

\begin{figure}
    \centering
    \input{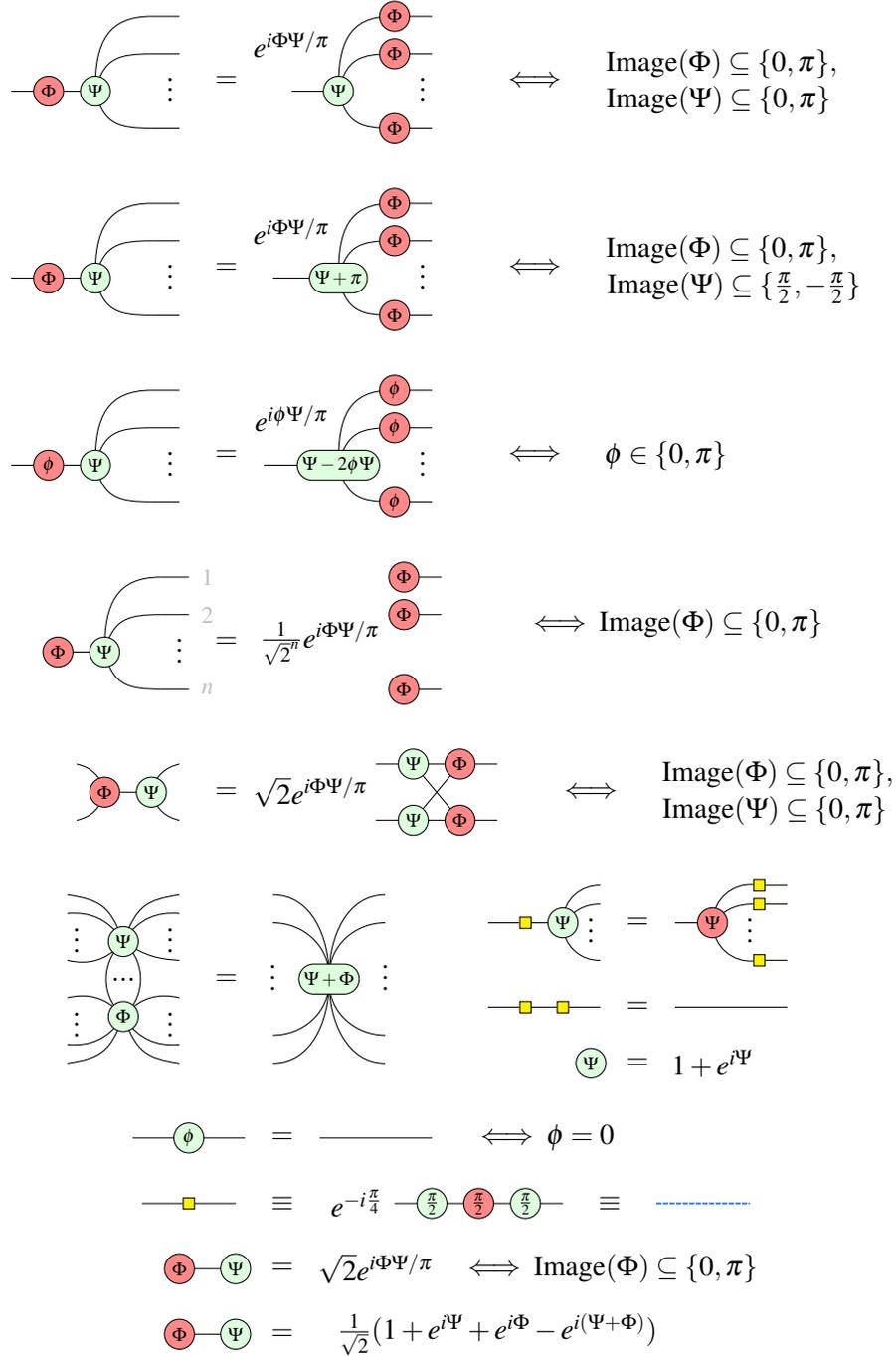}
    \caption{The complete set of polar-parameterised rewriting rules, where $\Psi$ and $\Phi$ are parameterised or fixed phases such that $\image(\Psi)\subseteq\{\alpha,\alpha+\pi\}$ and $\image(\Phi)\subseteq\{\beta,\beta+\pi\}$, given $\alpha,\beta\in\mathbb{R}$. Meanwhile, $\phi\in\mathbb{R}$ denotes a strictly fixed (non-parametric) phase.}
    \label{fig:zx-rules-param}
\end{figure}

Close inspection of these rules will reveal that, for Clifford ZX-diagrams, they represent perfect generalisations\footnote{Strictly speaking, the identity removal rule is an exception. Since this rule requires a specific phase of $0$, it is not possible to generalise this to polar-parameterised phases. However, this is not a problem in practice as such instances may be pushed to one side via polar-parameterised $\pi$-commutation.} of the non-parameterised rules of Figure \ref{fig:basicrules}. However, the same is not quite true of the broader Clifford+T gateset. Specifically, the $\pi$-commutation rule cannot be fully parameterised while maintaining generality. Particular note should be made of the fact that the $\pi$-commutation rule has been divided into three distinct polar-parameterised rules. These represent the three distinct cases in which (from Figure \ref{fig:basicrules}) the mapping $\Psi\rightarrow\Psi-2\Psi\Phi$ is expressible as a mapping from one polar-parameterised phase to another. Notably, parameterising this rule for T-like phases is not possible without breaking the parametric symmetry. This is outlined in Lemma \ref{lemma:no-T-pi-com}. As a result, constant (i.e. non-parameterised) Clifford+T ZX-diagrams are able to undergo further simplification than their equivalent polar-parameterised ZX-diagrams. This means the latter results in more stabiliser terms than the former. However, as will be seen, the difference in practice is generally very small.

\begin{lemma}
    \label{lemma:no-T-pi-com}
    Commuting a polarised phase $\Phi$, where $\image(\Phi)=\{0,\pi\}$, through a T-like spider (parameterised or otherwise) breaks parametric symmetry.
\end{lemma}
\begin{proof}
    $\pi$-commutation with a polar-parameterised phase $\Phi$, where $\image(\Phi)=\{0,\pi\}$, performs the following mapping to a separate phase: $\Psi\rightarrow\Psi-2\Psi\Phi$. If $\Psi$ is T-like (whether fixed or polar-parameterised), such as $\Psi=\frac{\pi}{4}$, then this mapping produces a phase of $\frac{\pi}{4}-\frac{\pi}{2}\Phi$. With a fractional coefficient to the polarised parameter, this phase is no longer polarised:
    \begin{equation*}
        \image\left(\frac{\pi}{4}-\frac{\pi}{2}\Phi\right)=\left\{\frac{\pi}{4},-\frac{\pi}{4}\right\}\neq\left\{\alpha,\alpha+\pi\right\}
    \end{equation*}
    for any $\alpha\in\mathbb{R}$. Hence, this gives rise to phases which are parameterised but not polarised, which, by definition, breaks the parametric symmetry. Lemma \ref{lemma:no-ambig-rule} in Appendix \ref{app:no-ambig-rule} illustrates the problems that would arise if such phases were permitted.
\end{proof}

\subsection{Parametric scalar expressions}

Reducing a polar-parameterised Clifford+T ZX-diagram via the rules of Figure \ref{fig:zx-rules-param} and stabiliser decomposition results in a parametric scalar expression, $S$, of the form:
\begin{equation}
    \label{eqn:sumprod}
    S=\sum\limits_{i=1}^{m} \Bigg[ C_i\prod\limits_{j=1}^{n_i} S_{ij} \Bigg]
\end{equation}

Here, $m$ is the number of stabiliser terms, $m\leq2^{\alpha t}$, given an initial T-count $t$. $C_i\in\mathbb{C}$ is then a global constant factor of term $i$ and $n_{ij}$ the number of parametric \textit{subterms} within term $i$. Lastly, $S_{ij}$ denotes the $j$\textsuperscript{th} subterm within the $i$\textsuperscript{th} term.

It can be shown, as in Appendix \ref{app:subterm-equiv}, that every subterm may be expressed in the form:

\begin{equation}
    \label{eqn:param-scal-expr}
    S_{ij}\;\;=\;\;1+e^{i\Psi_{ij}}+e^{i\Phi_{ij}}-e^{i(\Psi_{ij}+\Phi_{ij})}\;\;=\;\;\frac{1}{\sqrt{2}}\;\begin{tikzpicture}
	\begin{pgfonlayer}{nodelayer}
		\node [style=Z phase dot] (0) at (0, 0) {$\;\Psi_{ij}\;$};
		\node [style=X phase dot] (1) at (1.75, 0) {$\;\Phi_{ij}\;$};
	\end{pgfonlayer}
	\begin{pgfonlayer}{edgelayer}
		\draw (0) to (1);
	\end{pgfonlayer}
\end{tikzpicture}
 \;\quad\forall i,j
\end{equation}

where:

\begin{equation}
    \begin{aligned}
        \Psi_{ij}&=\alpha_{ij}+\pi\bigoplus\limits_{\mathfrak{p}\in P_{ij}^{\psi}}\mathfrak{p}\\
        \Phi_{ij}&=\beta_{ij}+\pi\bigoplus\limits_{\mathfrak{p}\in P_{ij}^{\phi}}\mathfrak{p}
    \end{aligned}
\end{equation}

with $\alpha_{ij},\beta_{ij}\in\mathbb{R}\;\forall i,j$ and:

\begin{equation}
    P_{ij}^{\psi},P_{ij}^{\phi}\subseteq\{\mathfrak{p}_1,\mathfrak{p}_2,\ldots,\mathfrak{p}_n\} \quad \forall i,j
\end{equation}

Here, $\{\mathfrak{p}_1,\mathfrak{p}_2,\ldots,\mathfrak{p}_n\}$ denotes the full set of uninstantiated boolean parameters from the initial parameterised ZX-diagram.

To summarise so far, any parametrically symmetric set of ZX-diagrams may be expressed as a single polar-parameterised ZX-diagram and reduced to a parameterised scalar expression. This expression may then be evaluated for any particular bitstring input, $\{\mathfrak{p}_1,\mathfrak{p}_2,\ldots,\mathfrak{p}_n\}\rightarrow\{p_1,p_2,\ldots,p_n\}$, where $p_i\in\mathbb{B}\;\forall i$, to deduce the scalar equivalent to having reduced the corresponding ZX-diagram. As a result, this method involves just a single instance of ZX-calculus reduction to scalar via stabiliser decomposition, followed by $2^n$ \textit{evaluations} of the resulting parametric scalar to deduce all possible outcomes. This is as opposed to the traditional method, which would instead require $2^n$ full (and potentially very slow) reductions to scalar of independent Clifford+T ZX-diagrams.

In practice, little (if any) speedup is achieved at this stage, since a single evaluation of the (potentially very large) parameterised scalar tends to be approximately as slow as simply reducing the corresponding ZX-diagram to scalar via stabiliser decomposition. However, the benefit lies in the fact that the parametric scalar expression always adheres to a strict consistent format and requires only very simple arithmetic to evaluate. As a result, it lends itself well to GPU-parallelism.

\subsection{GPU-parallelised evaluation}

To prepare it for the GPU, the parametric scalar attained from reducing a polar-parametric ZX-diagram may be expressed in a very primitive data structure. Specifically, each subterm, $S_{ij}$, may be recorded as two bitstrings to denote which parameters among the global set are included within its $P_{ij}^{\psi}$ and $P_{ij}^{\phi}$, together with two integers $a,b$ in the range $a,b\in[0,7]$ to denote $\alpha_{ij}=a\frac{\pi}{4}$ and $\beta_{ij}=b\frac{\pi}{4}$. Each subterm may thus be expressed as a successive row with the following headers, where, for instance, $\mathfrak{p}_l^{\psi}=1$ if $\mathfrak{p}_l\in P_{ij}^{\psi}$ and $\mathfrak{p}_l^{\psi}=1$ otherwise:

\begin{table}[!h]
    \centering
    \begin{tabular}{>{\centering\arraybackslash}p{0.75cm} | >{\centering\arraybackslash}p{0.75cm} >{\centering\arraybackslash}p{0.75cm} >{\centering\arraybackslash}p{0.75cm} >{\centering\arraybackslash}p{0.75cm} >{\centering\arraybackslash}p{0.75cm} | >{\centering\arraybackslash}p{0.75cm} >{\centering\arraybackslash}p{0.75cm} >{\centering\arraybackslash}p{0.75cm} >{\centering\arraybackslash}p{0.75cm} >{\centering\arraybackslash}p{0.75cm} }
    \hline
        \textbf{$*$} & 
        \textbf{$4\alpha/\pi$} & 
        \textbf{$\mathfrak{p}_1^{\psi}$} & 
        \textbf{$\mathfrak{p}_2^{\psi}$} & 
        \textbf{$\mathfrak{p}_3^{\psi}$} & 
        \textbf{$\cdots$} &
        \textbf{$4\beta/\pi$} & 
        \textbf{$\mathfrak{p}_1^{\phi}$} & 
        \textbf{$\mathfrak{p}_2^{\phi}$} & 
        \textbf{$\mathfrak{p}_3^{\phi}$} & 
        \textbf{$\cdots$} \\ \hline
    \end{tabular}
\end{table}

To optimise the indexing speed, it helps to artificially enforce each term, $i$, to share a consistent number, $n_i$, of subterms: $n_i=\max\limits_i(n_i)$. This can be achieved simply by padding `dummy' subterms into each term, with an additional bit, `*', to record whether each subterm is genuine (1) or a dummy (0). Additionally, the data should be stored in column-major order for a more efficient data access pattern, as detailed in Appendix \ref{app:coalescing}.

With the data now prepared into a primitive data structure, it may be evaluated with the aid of GPU parallelism. For any particular bitstring, $\{\mathfrak{p}_1,\mathfrak{p}_2,\ldots,\mathfrak{p}_n\}\rightarrow\{p_1,p_2,\ldots,p_n\}$ where $p_i\in\mathbb{B}\;\forall i$, each subterm may then be evaluated in parallel on the GPU with only simple arithmetic, consistent across all threads. The process of evaluating each subterm is as follows:

\begin{enumerate}
    \item If the row is marked as a dummy (i.e. the `$*$' column is $0$) then skip to step \ref{enum:paramzx:dummyStep}. In such cases, the row can essentially be ignored and its subterm result set immediately to $1$ without needing to compute any of the following calculations or steps. (This means the processing of dummy rows essentially amounts to simply doing nothing and waiting for the non-dummy threads to be processed.)
    
    \item\label{enum:paramzx:subStep} Substitute into each subterm expression (i.e. each row), in parallel, the chosen parameter values. This amounts to a simple bitwise multiplication (or AND operation):
    
    \begin{table}[!h]
        \centering
        \begin{tabular}{>{\centering\arraybackslash}p{0.75cm} | >{\centering\arraybackslash}p{0.75cm} >{\centering\arraybackslash}p{0.75cm} >{\centering\arraybackslash}p{0.75cm} >{\centering\arraybackslash}p{0.75cm} >{\centering\arraybackslash}p{0.75cm} | >{\centering\arraybackslash}p{0.75cm} >{\centering\arraybackslash}p{0.75cm} >{\centering\arraybackslash}p{0.75cm} >{\centering\arraybackslash}p{0.75cm} >{\centering\arraybackslash}p{0.75cm} }
        \hline
            \textbf{$1$} & 
            \textbf{$1$} & 
            \textbf{$p_1$} & 
            \textbf{$p_2$} & 
            \textbf{$p_3$} & 
            \textbf{$\cdots$} &
            \textbf{$1$} & 
            \textbf{$p_1$} & 
            \textbf{$p_2$} & 
            \textbf{$p_3$} & 
            \textbf{$\cdots$} \\ \hline
        \end{tabular}
    
        \vspace{0.25cm}
        $\times$
        \vspace{0.25cm}
        
        \begin{tabular}{>{\centering\arraybackslash}p{0.75cm} | >{\centering\arraybackslash}p{0.75cm} >{\centering\arraybackslash}p{0.75cm} >{\centering\arraybackslash}p{0.75cm} >{\centering\arraybackslash}p{0.75cm} >{\centering\arraybackslash}p{0.75cm} | >{\centering\arraybackslash}p{0.75cm} >{\centering\arraybackslash}p{0.75cm} >{\centering\arraybackslash}p{0.75cm} >{\centering\arraybackslash}p{0.75cm} >{\centering\arraybackslash}p{0.75cm} }
        \hline\hline
            \textbf{$*$} & 
            \textbf{$4\alpha/\pi$} & 
            \textbf{$\mathfrak{p}_1^\psi$} & 
            \textbf{$\mathfrak{p}_2^\psi$} & 
            \textbf{$\mathfrak{p}_3^\psi$} & 
            \textbf{$\cdots$} &
            \textbf{$4\beta/\pi$} & 
            \textbf{$\mathfrak{p}_1^\phi$} & 
            \textbf{$\mathfrak{p}_2^\phi$} & 
            \textbf{$\mathfrak{p}_3^\phi$} & 
            \textbf{$\cdots$} \\ \hline\hline
        \end{tabular}

        \vspace{0.25cm}
        $=$
        \vspace{0.25cm}
        
        \begin{tabular}{>{\centering\arraybackslash}p{0.75cm} | >{\centering\arraybackslash}p{0.75cm} >{\centering\arraybackslash}p{0.75cm} >{\centering\arraybackslash}p{0.75cm} >{\centering\arraybackslash}p{0.75cm} >{\centering\arraybackslash}p{0.75cm} | >{\centering\arraybackslash}p{0.75cm} >{\centering\arraybackslash}p{0.75cm} >{\centering\arraybackslash}p{0.75cm} >{\centering\arraybackslash}p{0.75cm} >{\centering\arraybackslash}p{0.75cm} }
        \hline\hline
            \textbf{$*$} & 
            \textbf{$4\alpha/\pi$} & 
            \textbf{$\mathfrak{p}_1^\psi p_1$} & 
            \textbf{$\mathfrak{p}_2^\psi p_2$} & 
            \textbf{$\mathfrak{p}_3^\psi p_3$} & 
            \textbf{$\cdots$} &
            \textbf{$4\beta/\pi$} & 
            \textbf{$\mathfrak{p}_1^\phi p_1$} & 
            \textbf{$\mathfrak{p}_2^\phi p_2$} & 
            \textbf{$\mathfrak{p}_3^\phi p_3$} & 
            \textbf{$\cdots$} \\ \hline\hline
        \end{tabular}
    \end{table}

    \item\label{enum:paramzx:xorStep} Calculate the XOR strings within $\Psi=(\mathfrak{p}_1^\psi p_1\oplus\mathfrak{p}_2^\psi p_2\oplus\ldots\oplus\mathfrak{p}_n^\psi p_n)\pi+\alpha$ and $\Phi=(\mathfrak{p}_1^\phi p_1\oplus\mathfrak{p}_2^\phi p_2\oplus\ldots\oplus\mathfrak{p}_n^\phi p_n)\pi+\beta$. That is, reduce $(\mathfrak{p}_1^\psi p_1\oplus\mathfrak{p}_2^\psi p_2\oplus\ldots\oplus\mathfrak{p}_n^\psi p_n)\rightarrow x$ and $(\mathfrak{p}_1^\phi p_1\oplus\mathfrak{p}_2^\phi p_2\oplus\ldots\oplus\mathfrak{p}_n^\phi p_n)\rightarrow y$, where $x,y\in\mathbb{B}$. This can be computed for each row in parallel, relying only on basic arithmetic:

    \begin{table}[!h]
        \centering
        \begin{tabular}{>{\centering\arraybackslash}p{0.75cm} | >{\centering\arraybackslash}p{0.75cm} >{\centering\arraybackslash}p{0.75cm} >{\centering\arraybackslash}p{0.75cm} >{\centering\arraybackslash}p{0.75cm} >{\centering\arraybackslash}p{0.75cm} | >{\centering\arraybackslash}p{0.75cm} >{\centering\arraybackslash}p{0.75cm} >{\centering\arraybackslash}p{0.75cm} >{\centering\arraybackslash}p{0.75cm} >{\centering\arraybackslash}p{0.75cm} }
        \hline\hline
            \textbf{$*$} & 
            \textbf{$4\alpha/\pi$} & 
            \textbf{$\mathfrak{p}_1^\psi p_1$} & 
            \textbf{$\mathfrak{p}_2^\psi p_2$} & 
            \textbf{$\mathfrak{p}_3^\psi p_3$} & 
            \textbf{$\cdots$} &
            \textbf{$4\beta/\pi$} & 
            \textbf{$\mathfrak{p}_1^\phi p_1$} & 
            \textbf{$\mathfrak{p}_2^\phi p_2$} & 
            \textbf{$\mathfrak{p}_3^\phi p_3$} & 
            \textbf{$\cdots$} \\ \hline\hline
        \end{tabular}

        \vspace{0.25cm}
        $\downarrow$
        \vspace{0.25cm}
        
        \begin{tabular}{>{\centering\arraybackslash}p{0.75cm} | >{\centering\arraybackslash}p{0.75cm} >{\centering\arraybackslash}p{0.75cm} >{\centering\arraybackslash}p{0.75cm} >{\centering\arraybackslash}p{0.75cm} >{\centering\arraybackslash}p{0.75cm} | >{\centering\arraybackslash}p{0.75cm} >{\centering\arraybackslash}p{0.75cm} >{\centering\arraybackslash}p{0.75cm} >{\centering\arraybackslash}p{0.75cm} >{\centering\arraybackslash}p{0.75cm} }
        \hline\hline
            \textbf{$*$} & 
            \textbf{$4\alpha/\pi$} & 
            & 
            \textbf{$x$} & 
            & 
            &
            \textbf{$4\beta/\pi$} & 
            & 
            \textbf{$y$} & 
            & 
            \\ \hline\hline
        \end{tabular}
    \end{table}

    \item\label{enum:paramzx:addConst} Given $\Psi=x\pi+\alpha \mod 2\pi$ and $\Phi=y\pi+\beta \mod 2\pi$, calculate $\frac{4\Psi}{\pi}=4x+\frac{4\alpha}{\pi} \mod 8$ and $\frac{4\Phi}{\pi}=4y+\frac{4\beta}{\pi} \mod 8$, such that $\frac{4\Psi}{\pi},\frac{4\Phi}{\pi}\in\{0,1,2,\ldots,7\}$. This too can be calculated for each row in parallel, using only simple operations on integers:

    \begin{table}[!h]
        \centering
        \begin{tabular}{>{\centering\arraybackslash}p{0.75cm} | >{\centering\arraybackslash}p{0.75cm} >{\centering\arraybackslash}p{0.75cm} >{\centering\arraybackslash}p{0.75cm} >{\centering\arraybackslash}p{0.75cm} >{\centering\arraybackslash}p{0.75cm} | >{\centering\arraybackslash}p{0.75cm} >{\centering\arraybackslash}p{0.75cm} >{\centering\arraybackslash}p{0.75cm} >{\centering\arraybackslash}p{0.75cm} >{\centering\arraybackslash}p{0.75cm} }
        \hline\hline
            \textbf{$*$} & 
            \textbf{$4\alpha/\pi$} & 
            & 
            \textbf{$x$} & 
            & 
            &
            \textbf{$4\beta/\pi$} & 
            & 
            \textbf{$y$} & 
            & 
            \\ \hline\hline
        \end{tabular}

        \vspace{0.25cm}
        $\downarrow$
        \vspace{0.25cm}
        
        \begin{tabular}{>{\centering\arraybackslash}p{0.75cm} | >{\centering\arraybackslash}p{0.75cm} >{\centering\arraybackslash}p{0.75cm} >{\centering\arraybackslash}p{0.75cm} >{\centering\arraybackslash}p{0.75cm} >{\centering\arraybackslash}p{0.75cm} | >{\centering\arraybackslash}p{0.75cm} >{\centering\arraybackslash}p{0.75cm} >{\centering\arraybackslash}p{0.75cm} >{\centering\arraybackslash}p{0.75cm} >{\centering\arraybackslash}p{0.75cm} }
        \hline\hline
            \textbf{$*$} & 
            & 
            & 
            \textbf{$4\Psi/\pi$} & 
            & 
            &
            & 
            & 
            \textbf{$4\Phi/\pi$} & 
            & 
            \\ \hline\hline
        \end{tabular}
    \end{table}

    \item\label{enum:paramzx:calcExp} The next step is to calculate $e^{i\Psi}$ and $e^{i\Phi}$. However, computing complex exponentials, such as these, is computationally costly. So, to avoid this, a lookup table may be used instead, mapping $4\Psi/\pi\in\{0,1,2,\ldots,7\}$ to $e^{i\Psi}\equiv A_\psi+B_\psi\sqrt{2}+i(C_\psi+D_\psi\sqrt{2})$, where this form is used to record the complex numbers as a set of simple fractional numbers.


    The same method may be used to compute $e^{i\Phi}\equiv A_\phi+B_\phi\sqrt{2}+i(C_\phi+D_\phi\sqrt{2})$ from $4\Phi/\pi$, and indeed to compute $e^{i(\Psi+\Phi)}\equiv A_{\psi+\phi}+B_{\psi+\phi}\sqrt{2}+i(C_{\psi+\phi}+D_{\psi+\phi}\sqrt{2})$ from $\frac{4\Psi}{\pi}+\frac{4\Phi}{\pi} \mod 8$. As ever, these calculations can be computed for each row in parallel:
    \newpage
    \begin{table}[!h]
        \centering
        \begin{tabular}{>{\centering\arraybackslash}p{0.75cm} | >{\centering\arraybackslash}p{0.75cm} >{\centering\arraybackslash}p{0.75cm} >{\centering\arraybackslash}p{0.75cm} >{\centering\arraybackslash}p{0.75cm} >{\centering\arraybackslash}p{0.75cm} | >{\centering\arraybackslash}p{0.75cm} >{\centering\arraybackslash}p{0.75cm} >{\centering\arraybackslash}p{0.75cm} >{\centering\arraybackslash}p{0.75cm} >{\centering\arraybackslash}p{0.75cm} }
        \hline\hline
            \textbf{$*$} & 
            & 
            & 
            \textbf{$4\Psi/\pi$} & 
            & 
            &
            & 
            & 
            \textbf{$4\Phi/\pi$} & 
            & 
            \\ \hline\hline
        \end{tabular}

        \vspace{0.25cm}
        $\downarrow$
        \vspace{0.25cm}
        
        \begin{tabular}{>{\centering\arraybackslash}p{0.5cm} | >{\centering\arraybackslash}p{0.5cm} >{\centering\arraybackslash}p{0.5cm} >{\centering\arraybackslash}p{0.5cm} >{\centering\arraybackslash}p{0.5cm} | >{\centering\arraybackslash}p{0.5cm} >{\centering\arraybackslash}p{0.5cm} >{\centering\arraybackslash}p{0.5cm} >{\centering\arraybackslash}p{0.5cm} | >{\centering\arraybackslash}p{0.75cm} >{\centering\arraybackslash}p{0.75cm} >{\centering\arraybackslash}p{0.75cm} >{\centering\arraybackslash}p{0.75cm} }
        \hline\hline
            \textbf{$*$} & 
            \textbf{$A_\psi$} & 
            \textbf{$B_\psi$} & 
            \textbf{$C_\psi$} & 
            \textbf{$D_\psi$} & 
            \textbf{$A_\phi$} &
            \textbf{$B_\phi$} & 
            \textbf{$C_\phi$} & 
            \textbf{$D_\phi$} &
            \textbf{$A_{\psi+\phi}$} &
            \textbf{$B_{\psi+\phi}$} & 
            \textbf{$C_{\psi+\phi}$} & 
            \textbf{$D_{\psi+\phi}$}
            \\ \hline\hline
        \end{tabular}
    \end{table}

    \item\label{enum:paramzx:calcSubterm} Lastly, having computed $e^{i\Psi}$, $e^{i\Phi}$, and $e^{i(\Psi+\Phi)}$ for each row, $ij$, each final subterm result may be deduced in parallel by calculating $s_{ij}=1+e^{i\Psi}+e^{i\Phi}-e^{i(\Psi+\Phi)}$. The result, in each case, may be stored as four simple fractional numbers via the form $s_{ij}\equiv A+B\sqrt{2}+i(C+D\sqrt{2})$. Note that adding two such complex numbers, $e^{i\Psi}+e^{i\Phi}$, in this form amounts to adding the like terms, as outlined in Appendix \ref{app:para-sum-alg}.

    Now, each row will store four simple numbers to record its subterm scalar:
    \begin{table}[!h]
        \centering
        \begin{tabular}{>{\centering\arraybackslash}p{0.5cm} | >{\centering\arraybackslash}p{0.5cm} >{\centering\arraybackslash}p{0.5cm} >{\centering\arraybackslash}p{0.5cm} >{\centering\arraybackslash}p{0.5cm} | >{\centering\arraybackslash}p{0.5cm} >{\centering\arraybackslash}p{0.5cm} >{\centering\arraybackslash}p{0.5cm} >{\centering\arraybackslash}p{0.5cm} | >{\centering\arraybackslash}p{0.75cm} >{\centering\arraybackslash}p{0.75cm} >{\centering\arraybackslash}p{0.75cm} >{\centering\arraybackslash}p{0.75cm} }
        \hline\hline
            \textbf{$*$} & 
            \textbf{$A_\psi$} & 
            \textbf{$B_\psi$} & 
            \textbf{$C_\psi$} & 
            \textbf{$D_\psi$} & 
            \textbf{$A_\phi$} &
            \textbf{$B_\phi$} & 
            \textbf{$C_\phi$} & 
            \textbf{$D_\phi$} &
            \textbf{$A_{\psi+\phi}$} &
            \textbf{$B_{\psi+\phi}$} & 
            \textbf{$C_{\psi+\phi}$} & 
            \textbf{$D_{\psi+\phi}$}
            \\ \hline\hline
        \end{tabular}

        \vspace{0.25cm}
        $\downarrow$
        \vspace{0.25cm}
        
        \begin{tabular}{>{\centering\arraybackslash}p{0.75cm} | >{\centering\arraybackslash}p{2.25cm} >{\centering\arraybackslash}p{2.25cm} >{\centering\arraybackslash}p{2.25cm} >{\centering\arraybackslash}p{2.25cm} }
        \hline\hline
            \textbf{$*$} & 
            \textbf{$A_{1,\psi,\phi,-(\psi+\phi)}$} & 
            \textbf{$B_{1,\psi,\phi,-(\psi+\phi)}$} & 
            \textbf{$C_{1,\psi,\phi,-(\psi+\phi)}$} & 
            \textbf{$D_{1,\psi,\phi,-(\psi+\phi)}$}
            \\ \hline\hline
        \end{tabular}
    \end{table}

    \item\label{enum:paramzx:dummyStep} Hiding the subscript labels here for brevity, the culmination of these steps is a matrix wherein each row now records four simple numbers which collectively denote a single subterm value:
    \begin{table}[!h]
        \centering
        \begin{tabular}{>{\centering\arraybackslash}p{2.25cm} >{\centering\arraybackslash}p{2.25cm} >{\centering\arraybackslash}p{2.25cm} >{\centering\arraybackslash}p{2.25cm} }
        \hline\hline
            \textbf{$A$} & 
            \textbf{$B$} & 
            \textbf{$C$} & 
            \textbf{$D$}
            \\ \hline\hline
        \end{tabular}
    \end{table}
    
    If the subterm was flagged as a dummy then one may jump straight to this step with $A=1,\;B=0,\;C=0,\;D=0$.
\end{enumerate}

Once every subterm has been fully reduced, it remains to multiply together all the subterms within each term, and then to sum together these results, as per Equation \ref{eqn:sumprod}. These two steps may make further use of GPU parallelism, utilising the algorithm outlined in Appendix \ref{app:para-sum-alg}. This whole approach may be used for repeated strong simulation or applied directly to the left-hand side of Figure \ref{fig:classicsim} for computing marginal probabilities. It may also be applied to repeated weak simulation as per Appendix \ref{app:repeated-weak}.

\section{Results}

An illustrative summary of this new parametric and GPU-parallelised method, as compared to the conventional approach of \cite{DBLP:journals/quant-ph/Kissinger21}, is shown in Figure \ref{fig:flows}.

\begin{figure}[h!]
    \begin{subfigure}[b]{0.48\textwidth} 
        \centering
        \includegraphics[scale=1.125]{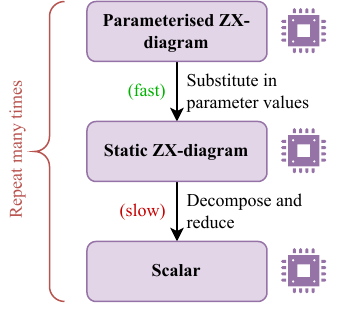}
        \captionsetup{width=0.9\textwidth}
        \caption{The conventional CPU-based approach, introduced in \cite{DBLP:journals/quant-ph/Kissinger21}.}
    \end{subfigure}
    \begin{subfigure}[b]{0.48\textwidth}
        \centering
        \includegraphics[scale=1.125]{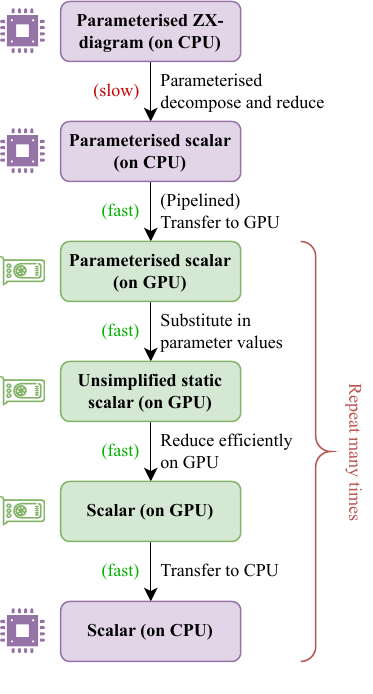}
        \captionsetup{width=0.9\textwidth}
        \caption{The parallelised GPU-based approach outlined in this chapter.}
    \end{subfigure}
    \caption{A comparison of the two procedures for repeated evaluation of a parameterised ZX-diagram for various sets of parameter values. The boxes show the state of the data at each step, connected by arrows indicating the processes that update these data, together with a qualitative note of the speed of each process. Purple boxes represent data being on the CPU (\raisebox{-0.2ex}{\includegraphics[height=1.8ex]{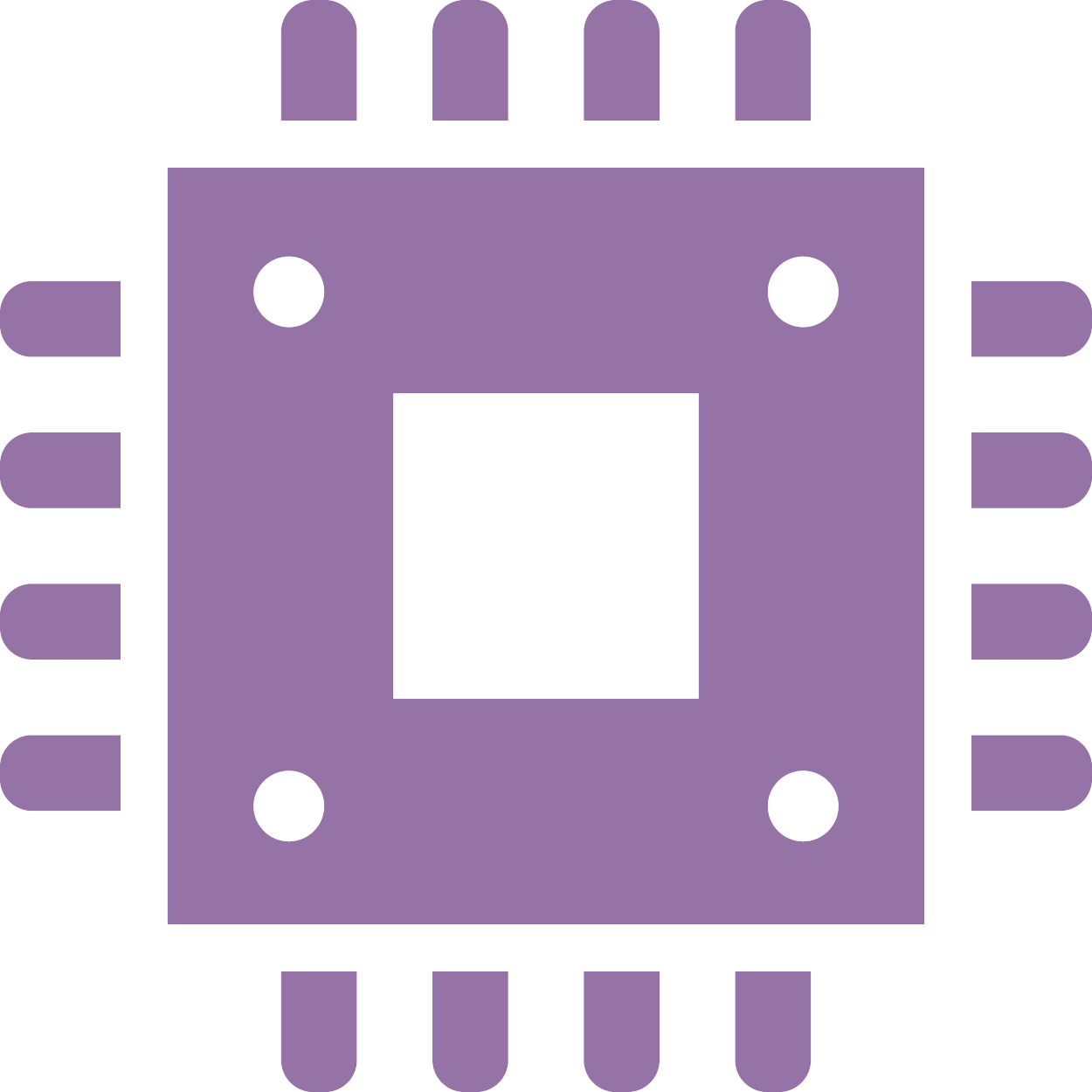}}) and green boxes data on the GPU (\raisebox{-0.2ex}{\includegraphics[height=1.8ex]{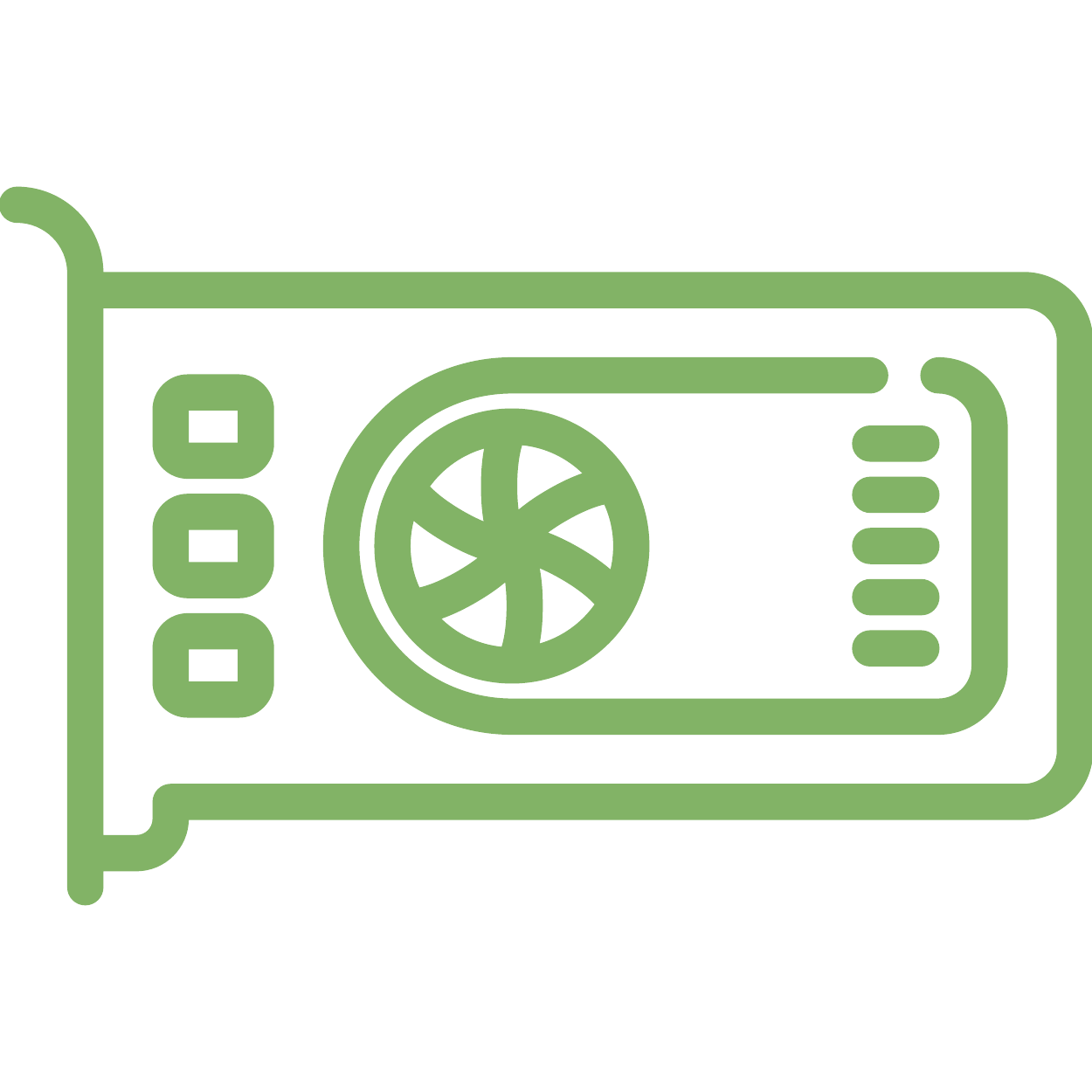}}).}
    \label{fig:flows}
\end{figure}

To measure the effectiveness of this approach, it was benchmarked against the conventional method of Kissinger and van de Wetering \cite{DBLP:journals/quant-ph/Kissinger21}, for strongly simulating randomly generated Clifford+T circuits of various T-counts and parameter counts. Specifically, we compare against the Quizx \cite{github:quizx} implementation of the conventional method, which is a Rust-based port of PyZX \cite{github:pyzx,kissinger2020Pyzx}, designed for speed and performance. Note that all experiments were run on relatively modest commercial hardware with the following specs: \textit{6-core 2.69GHz Intel i5-11400H CPU and an 8GB NVIDIA GeForce GTX 1650 GPU, plus 8GB SODIMM RAM}.

For each circuit, we measure the time taken by both methods to strongly simulate various instances (i.e. evaluate its scalar for particular parameter bitstrings). The speedup factor of the new method versus the old is then given by the ratio of these two runtimes. The speedup factor given a particular number $N$ of evaluations is labelled $S_N$.

The new method carries an initial overhead runtime to reduce the original parameterised ZX-diagram to a GPU-ready parameterised scalar, though every subsequent evaluation may be computed very rapidly. Consequently, as an increasing number of evaluations $N$ are taken, the initial overhead time becomes increasingly negligible and hence the speedup factor increases. However, as the number of parallel threads available on the GPU is finite, each experiment has a maximum `\textit{terminal}' speedup factor, which is approached asymptotically as $N\rightarrow\infty$.

Both theoretically and experimentally, the speedup factor $S_N$ thus follows a sigmoid curve (as pictured in Appendix \ref{app:sigmoid}) versus the number of evaluations $N$ when plotted on a log-linear scale. With this in mind, this sigmoid curve for any particular circuit may be described by two metrics, namely:
\begin{itemize}
    \item the \textbf{terminal speedup factor}, $S_\infty$, which measures the theoretical maximum runtime improvement for the given parameterised ZX-diagram, and
    \item the \textbf{inflection point}, which measures how many evaluations are required for the speedup factor to reach half of its theoretical maximum, and which effectively quantifies the rate at which the speedup factor increases against the number of evaluations taken.
\end{itemize}

\begin{figure}[h!]
    \centering
    \begin{subfigure}[t]{0.49\textwidth}
        \centering
        \includegraphics[scale=0.5]{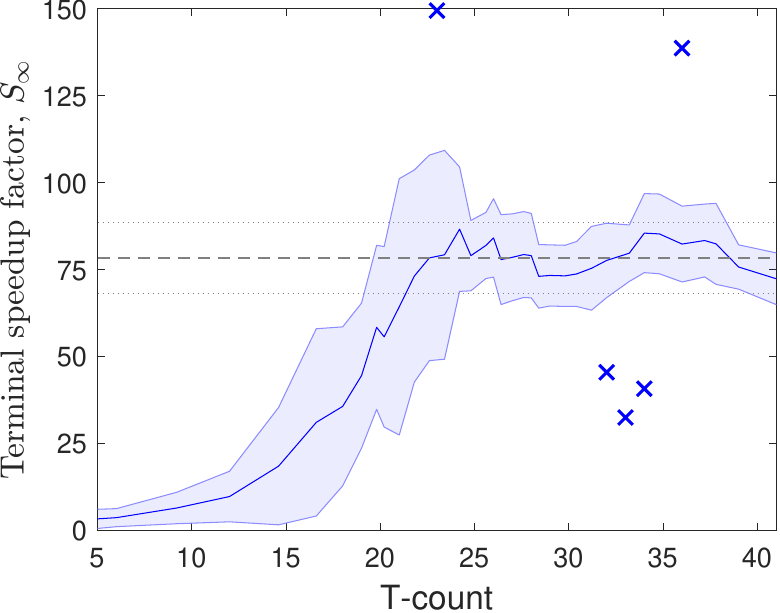}
        \captionsetup{width=0.9\textwidth}
        \caption{The terminal speedup factors, $S_\infty$, versus Quizx, plotted against T-count.}
        \label{fig:meas-tspeedups}
    \end{subfigure}%
    ~ 
    \begin{subfigure}[t]{0.49\textwidth}
        \centering
        \includegraphics[scale=0.5]{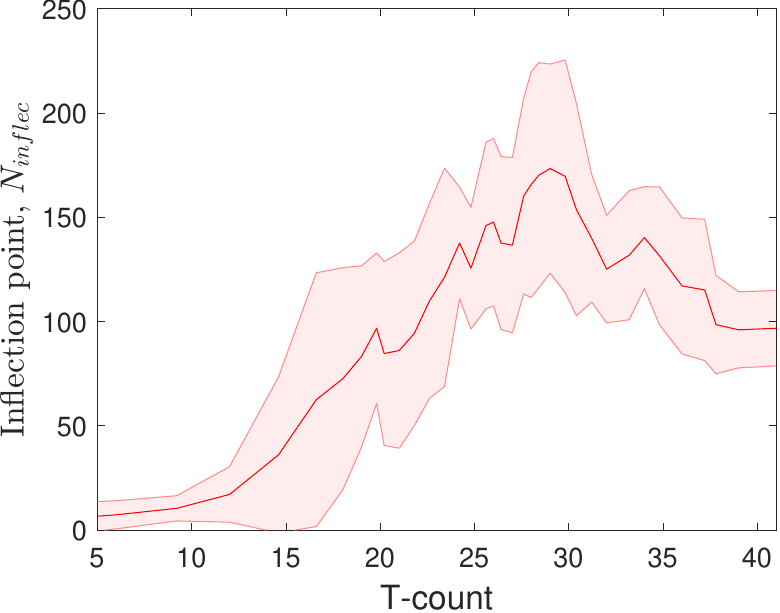}
        \captionsetup{width=0.9\textwidth}
        \caption{The inflection points, $N_{inflec}$, plotted against T-count.}
        \label{fig:meas-inflections}
    \end{subfigure}
    \caption{The measured (a) terminal speedup factors and (b) inflection points of various randomly generated parameterised ZX-diagrams versus the conventional Quizx implementation \protect\cite{DBLP:journals/quant-ph/Kissinger21}, plotted against the initial T-counts (after Clifford simplification). These plots show a ($5$-point window) moving average of the measured results, with error margins given by the standard deviation. Extreme outlier measurements have been excluded and indicated with an `x', and Figure (a) includes a dashed grey line indicating the average $S_\infty$ (with standard deviation error margins) across all (non-outlier) measurements of T-counts $\geq25$, where $S_\infty$ appears to plateau.}
    \label{fig:speedups}
\end{figure}

Figure \ref{fig:speedups} shows how these metrics varied against the initial (post Clifford simplification) T-counts of the various circuits. Evidently, for trivially small T-counts, little improvement is observed. This is unsurprising as such circuits decompose into very few stabiliser terms --- too few to take full advantage of GPU parallelism (especially if the number of terms is fewer than the number of available parallel threads).

Nevertheless, beyond the trivially small cases the measured terminal speedup factors appears to plateau, indicating that the full GPU capacity is being utilised. In this region (T-counts $\geq25$), we observe an average terminal speedup factor of $S_\infty=78.3\pm10.2$, with rare but extreme outliers as high as $149.5$ and as low as $32.4$. The lower outliers are believed to represent particularly unlucky ZX-diagrams for which the problem described in Lemma \ref{lemma:no-T-pi-com} was more prominent. In other words, in very unfortunate (though rare) cases, non-parametric Clifford+T ZX-diagrams were able to undergo much more simplification and produce significantly fewer stabiliser terms than their parametric counterparts. The very high outliers, meanwhile, are likely due to cases which simplified unusually neatly, producing very few subterms per term, easing the computational cost.

Evidently, therefore, this method is generally very effective at speeding up the problem of strong simulation by utilising GPU hardware, particularly when a larger number of evaluations (i.e. simulating more elements in a parametrically symmetric set of ZX-diagrams), as is very common in weak simulation tasks, where the number of evaluations is essentially exponential against the qubit count. From the inflection point measurements of Figure \ref{fig:meas-inflections}, it is clear that hundreds of measurements is sufficient to reach at least half of the full potential speedup in each case, with thousands then being sufficient to surpass $90\%$ of the terminal speedup factor. This means in practical situations speedup factors very near the $S_\infty=78.3\pm10.2$ mark are indeed viable.

All the relevant code is available at \url{https://github.com/mjsutcliffe99/ParamZX} \cite{github:paramzx}.

\section{Conclusions and Future Directions}

There are many applications in ZX-calculus where it is useful to compute (and often sum) many instances of the reduced scalar of a particular circuit, for various boolean input/output bitstrings. This paper highlights the redundancy inherent in such cases by demonstrating how the quantum circuit reduction strategies of ZX-calculus can be parameterised, with appropriate considerations, to allow circuits to be reduced while maintaining arbitrary boolean inputs/outputs. In effect this means that instead of reducing a given circuit $n$ times to compute its scalar for $n$ different input/output bitstrings, one can instead reduce it just once (while maintaining generality), and efficiently \textit{evaluate} the resulting parameterised expression $n$ times using GPU hardware. Ultimately, it was shown that, applied to classical simulation, this led to an average speedup factor of $78\pm10.2$.

While not implemented within the scope of this paper, there are a number of small techniques that could be employed to further optimise the method outlined here. For instance, after a given circuit has been reduced to a parameterised scalar, there may often be many simplifications that could be applied to minimise the number of subterms involved. For example, node-type subterms that share the same set of parameters but whose constant terms differ by $\pi$ may be cancelled pairwise in place of a constant, such as $(1+e^{i\pi(\frac{2}{4}+(a\oplus b))})(1+e^{i\pi(\frac{6}{4}+(a\oplus b))})=(1+e^{i\pi\frac{2}{4}})(1+e^{i\pi\frac{6}{4}})=(1+i)(1-i)=2$. 
A special case of this is in cancelling such pairs whose constants are $0$ and $\pi$ respectively, as this reduces overall to $0$, hence reducing the entire scalar term to $0$, negating any need to compute it further. This would increase the initial overhead time, and by extension the inflection point, but would reduce the number of subterms to compute and hence improve the terminal speedup factor.

Furthermore, this paper primarily focused on applying the method outlined within to the use-case of classical simulation via the summing and/or doubling approaches. However, the method is a very general one, applicable to many related applications such as breaking simulation tasks into disjoint components by vertex cutting~\cite{Codsi2022Masters,sutcliffe-partitioning,sutcliffePhd} and performing weak simulation via the metropolis method described in \cite{bravyi19}.

\nocite{*}
\bibliographystyle{eptcs}
\bibliography{param-zx}

\appendix

\newpage
\section{Preserving parametric symmetry}
\label{app:no-ambig-rule}

\begin{lemma}
    \label{lemma:no-ambig-rule}
    Reducing a parametrically symmetric set of ZX-diagrams as a single polar-parameterised ZX-diagram requires every phase at every step to be either fixed or polarised.
\end{lemma}
\begin{proof}
    As a proof by example, consider the following parameterised ZX-diagram containing a non-polarised phase $\mathfrak{p}\frac{\pi}{2}$, where $\mathfrak{p}$ is an uninstantiated boolean parameter:
    \ctikzfig{illegal-lc}
    Here, the central parameterised phase is not polarised:
    \begin{equation*}
        \image\left(\mathfrak{p}\frac{\pi}{2}\right)=\left\{0,\frac{\pi}{2}\right\}\neq\left\{\alpha,\alpha+\pi\right\}
    \end{equation*}
    for any $\alpha\in\mathbb{R}$. Consequently, in this case it is ambiguous whether the local complementation rule applies, as it would do so when $\mathfrak{p}\rightarrow1$ but not when $\mathfrak{p}\rightarrow0$. In other words, the two possibilities of $\mathfrak{p}\rightarrow 0$ and $\mathfrak{p}\rightarrow 1$ would lead to two diverging branches, meaning the parametrically symmetric pair would no longer be expressible as a single parameterised ZX-diagram.
\end{proof}

\section{Additional polar-parameterised rewrite rules}
\label{app:more-rules}

In additional to the generalised rules given in Figure \ref{fig:zx-rules-param}, some higher-level rules (derivable from the basic set) may also be generalised for polar-parameterised ZX-diagrams. Specifically, these are parametric versions of the \textit{local complementation} and \textit{pivoting} rules \cite{DBLP:journals/quant-ph/Wetering20}:

\ctikzfig{deriv-rules-param}

The derivations of these parametric rules, including those of Figure \ref{fig:zx-rules-param}, follow from their non-parametric counterparts and are left as an exercise for the reader.

\section{Equivalence of subterm types}
\label{app:subterm-equiv}

From the parameterised rewriting rules and scalar relations presented in Figure \ref{fig:zx-rules-param} and Appendix \ref{app:more-rules}, one may observe that there are in fact four different types of parameterised subterms that may arise. These are summarised in table \ref{tab:paramzx:subtermtypes}.

\begin{table}[ht]
    \centering
    {
    \renewcommand{\arraystretch}{2}
    \begin{tabular}{>{\centering\arraybackslash}p{3cm} >{\centering\arraybackslash}p{5cm} >{\centering\arraybackslash}p{4cm}}
        \textbf{Name} & 
        \textbf{Form} & 
        \textbf{Origin(s)} \\ \hline\hline
        
        Node & 
        $(1+e^{i\Psi})$ & 
        \begin{tikzpicture}
	\begin{pgfonlayer}{nodelayer}
		\node [style=Z phase dot] (0) at (0, 0) {$\Psi$};
	\end{pgfonlayer}
\end{tikzpicture}  \\ \hline
        
        Phase-pair & 
        $\left(1+e^{i\Psi} + e^{i\Phi} - e^{i(\Psi+\Phi)}\right)$ & 
        \begin{tikzpicture}
	\begin{pgfonlayer}{nodelayer}
		\node [style=Z phase dot] (0) at (0, 0) {$\Psi$};
		\node [style=X phase dot] (1) at (1.5, 0) {$\Phi$};
	\end{pgfonlayer}
	\begin{pgfonlayer}{edgelayer}
		\draw (0) to (1);
	\end{pgfonlayer}
\end{tikzpicture}  \\ \hline
        
        Half-$\pi$ & 
        $e^{i\Psi/2}$ & 
        Appendix \ref{app:more-rules} \\ \hline
        
        $\pi$-pair & 
        $e^{i\Psi\Phi/\pi}$ & 
        Figure \ref{fig:zx-rules-param} \& Appendix \ref{app:more-rules}  \\ \hline
    \end{tabular}
    }
    \caption{The different types of subterms that may arise from reducing parameterised ZX-diagrams, where $\Psi$ and $\Phi$ are polar-parameterised phases obeying the form of Equation \ref{eqn:param-phase}.}
    \label{tab:paramzx:subtermtypes}
\end{table}

In fact, it can be shown that these may all reduce to a single unique subterm type, namely that labelled `\textit{phase-pair}'. Lemmas \ref{lemma:paramzx:node-equiv} to \ref{lemma:paramzx:halfpi-equiv} formalise this observation.

\begin{lemma}
    Node-type subterms can be reduced to phase-pair subterms:
    \begin{equation}
        \left(1+e^{i\Psi}\right)=C\left(1+e^{i\Psi'}+e^{i\Phi}-e^{i(\Psi' +\Phi)}\right)
    \end{equation}
    where $\Psi'=\Psi+\frac{\pi}{2}$ and $C=\frac{\sqrt{2}}{4}(1-i)$.
    \label{lemma:paramzx:node-equiv}
\end{lemma}
\begin{proof}
    \begin{equation}
        \begin{tikzpicture}
	\begin{pgfonlayer}{nodelayer}
		\node [style=none] (0) at (-0.75, 0) {$\left(1+e^{i\Psi}\right)$};
		\node [style=none] (1) at (1.5, 0) {$=$};
		\node [style=Z phase dot] (2) at (2.75, 0) {$\Psi$};
		\node [style=none] (3) at (1.5, -2) {$=$};
		\node [style=Z phase dot] (4) at (3.25, -2) {$\;\Psi+\frac{\pi}{2}\;$};
		\node [style=Z phase dot] (5) at (6.25, -2) {$\;-\frac{\pi}{2}\;$};
		\node [style=none] (6) at (1.5, -4) {$=$};
		\node [style=Z phase dot] (7) at (3.25, -4) {$\;\Psi+\frac{\pi}{2}\;$};
		\node [style=X phase dot] (8) at (6.25, -4) {$\;-\frac{\pi}{2}\;$};
		\node [style=hadamard] (9) at (5, -4) {};
		\node [style=none] (10) at (1.5, -6) {$=$};
		\node [style=Z phase dot] (11) at (5, -6) {$\;\Psi+\frac{\pi}{2}\;$};
		\node [style=X phase dot] (12) at (10.25, -6) {$\;-\frac{\pi}{2}\;$};
		\node [style=none] (14) at (3, -6) {$e^{-i\frac{\pi}{4}}$};
		\node [style=X phase dot] (15) at (6.75, -6) {$\frac{\pi}{2}$};
		\node [style=Z phase dot] (16) at (7.75, -6) {$\frac{\pi}{2}$};
		\node [style=X phase dot] (17) at (8.75, -6) {$\frac{\pi}{2}$};
		\node [style=none] (18) at (1.5, -8) {$=$};
		\node [style=Z phase dot] (19) at (5, -8) {$\;\Psi+\frac{\pi}{2}\;$};
		\node [style=X dot] (20) at (8.75, -8) {};
		\node [style=none] (21) at (3, -8) {$e^{-i\frac{\pi}{4}}$};
		\node [style=X phase dot] (22) at (6.75, -8) {$\frac{\pi}{2}$};
		\node [style=Z phase dot] (23) at (7.75, -8) {$\frac{\pi}{2}$};
		\node [style=none] (24) at (1.5, -10) {$=$};
		\node [style=Z phase dot] (25) at (6, -10) {$\;\Psi+\frac{\pi}{2}\;$};
		\node [style=none] (27) at (3.5, -10) {$\frac{1}{\sqrt{2}}e^{-i\frac{\pi}{4}}$};
		\node [style=X phase dot] (28) at (7.75, -10) {$\frac{\pi}{2}$};
		\node [style=none] (29) at (1.5, -12) {$=$};
		\node [style=none] (31) at (10, -12) {$\frac{1}{\sqrt{2}}e^{-i\frac{\pi}{4}}\cdot \frac{1}{\sqrt{2}} \left(1+e^{i(\Psi+\frac{\pi}{2})}+e^{i\Phi}-e^{i(\Psi + \frac{\pi}{2} +\Phi)}\right)$};
		\node [style=none] (32) at (1.5, -14) {$=$};
		\node [style=none] (33) at (8.5, -14) {$\frac{\sqrt{2}}{4}(1-i)\left(1+e^{i\Psi'}+e^{i\Phi}-e^{i(\Psi' +\Phi)}\right)$};
	\end{pgfonlayer}
	\begin{pgfonlayer}{edgelayer}
		\draw (4) to (5);
		\draw (7) to (9);
		\draw (9) to (8);
		\draw (11) to (15);
		\draw (15) to (16);
		\draw (16) to (17);
		\draw (17) to (12);
		\draw (19) to (22);
		\draw (22) to (23);
		\draw (23) to (20);
		\draw (25) to (28);
	\end{pgfonlayer}
\end{tikzpicture}
    \end{equation}
\end{proof}

\begin{lemma}
    $\pi$-pair subterms can be reduced to phase-pair subterms:
    \begin{equation}
        e^{i\Psi\Phi/\pi} \rightarrow C\left(1+e^{i\Psi}+e^{i\Phi}-e^{i(\Psi +\Phi)}\right)
    \end{equation}
    where $C=\frac{1}{2}$.
    \label{lemma:paramzx:pipair-equiv}
\end{lemma}
\begin{proof}
    There are five sources from which $\pi$-pair subterms, $e^{i\Psi\Phi/\pi}$, may arise, being:
    \begin{itemize}
        \item parameterised state copy (Figure \ref{fig:zx-rules-param}),
        \item parameterised pivoting (Appendix \ref{app:more-rules}),
        \item parameterised $\pi$-commutation (Figure \ref{fig:zx-rules-param}),
        \item parameterised bialgebra (Figure \ref{fig:zx-rules-param}), and
        \item parameterised special case phase-pair (Figure \ref{fig:zx-rules-param}).
    \end{itemize}
    In each case, one may observe that $\image(\Phi)\subseteq\{0,\pi\}$. Under this restriction, the relation holds:
    \begin{equation}
        e^{i\Psi\Phi/\pi} = \frac{1}{2}\left(1+e^{i\Psi}+e^{i\Phi}-e^{i(\Psi +\Phi)}\right)
    \end{equation}
    Explicitly, if $\Phi=0$:
    \begin{align}
        e^0 &= \frac{1}{2}\left(1+e^{i\Psi}+e^0-e^{i(\Psi+0)}\right) \\
        1 &= \frac{1}{2}\left(1+1\right) \nonumber\\
        1 &= 1 \nonumber\\
        \therefore \text{LHS} &= \text{RHS} \nonumber
    \end{align}
    Likewise, if $\Phi=\pi$:
    \begin{align}
        e^{i\Psi} &= \frac{1}{2}\left(1+e^{i\Psi}+e^{i\pi}-e^{i(\Psi+\pi)}\right) \\
        e^{i\Psi} &= \frac{1}{2}\left(1+e^{i\Psi}+(-1)-e^{i\Psi}e^{i\pi}\right) \nonumber\\
        e^{i\Psi} &= \frac{1}{2}\left(1+e^{i\Psi}-1+e^{i\Psi}\right) \nonumber\\
        e^{i\Psi} &= \frac{1}{2}\left(2e^{i\Psi}\right) \nonumber\\
        e^{i\Psi} &= e^{i\Psi} \nonumber\\
        \therefore \text{LHS} &= \text{RHS} \nonumber
    \end{align}
    Hence, graphically:
    \begin{equation}
        \begin{tikzpicture}
	\begin{pgfonlayer}{nodelayer}
		\node [style=Z phase dot] (0) at (-0.5, 0) {$\Psi$};
		\node [style=X phase dot] (1) at (1, 0) {$\Phi$};
		\node [style=none] (2) at (2.5, 0) {$=$};
		\node [style=none] (3) at (5, 0) {$\sqrt{2}e^{i\Psi\Phi/\pi}$};
	\end{pgfonlayer}
	\begin{pgfonlayer}{edgelayer}
		\draw (0) to (1);
	\end{pgfonlayer}
\end{tikzpicture}

    \end{equation}
    provided $\image(\Phi)\subseteq\{0,\pi\}$.
\end{proof}

\begin{lemma}
    Half-$\pi$ subterms can be reduced to phase-pair subterms:
    \begin{equation}
        e^{i\Psi/2} \rightarrow C\left(1+e^{i\Psi'}+e^{i\Phi}-e^{i(\Psi' +\Phi)}\right)
    \end{equation}
    where $\Psi'=-\Psi+\frac{\pi}{2}$ and $C=\frac{1}{2}$.
    \label{lemma:paramzx:halfpi-equiv}
\end{lemma}
\begin{proof}
    Half-$\pi$ subterms arise from instances of parameterised local complementation (Appendix \ref{app:more-rules}), from which it may be observed that $\image(\Psi)\subseteq\{\frac{\pi}{2},\frac{3\pi}{2}\}$.

    These terms may be slightly rewritten with a change of variable:
    \begin{equation}
        e^{- i\frac{\Psi}{2}} e^{i\frac{\pi}{2}} = e^{\frac{i}{2}(-\Psi+\frac{\pi}{2})} e^{\frac{i\pi}{4}} = e^{\frac{i}{2}\Psi'} e^{\frac{i\pi}{4}}
    \end{equation}
    where $\Psi'=-\Psi+\frac{\pi}{2}$ such that $\image(\Psi')\subseteq\{0,\pi\}$. Furthermore:
    \begin{equation}
        e^{i\Psi'/2} \equiv e^{i\Psi'\Phi/\pi}
    \end{equation}
    where $\Phi=\frac{\pi}{2}$. Hence, the half-$\pi$ subterm type is a special case of the $\pi$-pair type, which was shown in Lemma \ref{lemma:paramzx:pipair-equiv} to be reducible to the phase-pair type.
\end{proof}

Given this, all such parameterised scalar expressions are expressible in a consistent format, described by Equation \ref{eqn:param-scal-expr}.

\section{Coalescing the data}
\label{app:coalescing}

Consider an example parametric scalar, in matrix form, consisting of $4$ parameters ($\mathfrak{p}_1,\mathfrak{p}_2,\mathfrak{p}_3,\mathfrak{p}_4$) and $2$ terms:

\begin{table}[!h]
    \centering
    \begin{tabular}{>{\centering\arraybackslash}p{0.75cm} | >{\centering\arraybackslash}p{0.75cm} >{\centering\arraybackslash}p{0.75cm} >{\centering\arraybackslash}p{0.75cm} >{\centering\arraybackslash}p{0.75cm} >{\centering\arraybackslash}p{0.75cm} | >{\centering\arraybackslash}p{0.75cm} >{\centering\arraybackslash}p{0.75cm} >{\centering\arraybackslash}p{0.75cm} >{\centering\arraybackslash}p{0.75cm} >{\centering\arraybackslash}p{0.75cm} }
    \hline
        \textbf{$*$} & 
        \textbf{$4\alpha/\pi$} & 
        \textbf{$\mathfrak{p}_1^\psi$} & 
        \textbf{$\mathfrak{p}_2^\psi$} & 
        \textbf{$\mathfrak{p}_3^\psi$} & 
        \textbf{$\mathfrak{p}_4^\psi$} &
        \textbf{$4\beta/\pi$} & 
        \textbf{$\mathfrak{p}_1^\phi$} & 
        \textbf{$\mathfrak{p}_2^\phi$} & 
        \textbf{$\mathfrak{p}_3^\phi$} & 
        \textbf{$\mathfrak{p}_4^\phi$} \\ \hline

        1 & 2 & 1 & 1 & 0 & 0 & 7 & 1 & 1 & 1 & 0 \\
        1 & 4 & 0 & 1 & 1 & 1 & 3 & 0 & 1 & 0 & 1 \\
        0 & 0 & 0 & 0 & 0 & 0 & 0 & 0 & 0 & 0 & 0 \\
        1 & 6 & 1 & 0 & 0 & 0 & 3 & 1 & 1 & 1 & 1 \\
        1 & 5 & 1 & 1 & 1 & 1 & 1 & 1 & 0 & 1 & 0 \\
        1 & 2 & 0 & 0 & 0 & 1 & 4 & 1 & 1 & 0 & 0 \\\hline
    \end{tabular}
\end{table}

While for practical purposes this data is treated as two-dimensional, as far as the memory is concerned it is in fact, necessarily stored linearly. Conventionally, this would be stored, linearised, in row-major order (storing the cells of the first row, followed by those of the second row and so on). This is appropriate when the data is to be processed row-by-row on the CPU (whose processing pattern is inherently sequential), as each subsequent element of data to be processed is stored immediately after the previous, ensuring convenient (quicker) access.

However, if instead, as in this case, each row of this data is to be processed in parallel on the GPU, then this memory arrangement proves to be suboptimal. In this scenario, it is preferential to store the data in column-major order (meaning all the cells of the first column are stored one after the other, followed then by the cells of the second column and so on). A justification for this follows:

\begin{lemma}
To better optimise for GPU processing, the two-dimensional subterm data should be stored in memory in column-major order (rather than the more typical row-major order that 2D arrays tend to be stored as).
\label{lemma:colmajor}

\end{lemma}

\begin{proof}
When the rows are to be processed in parallel, they will process their respective \textit{dummy flags} (their data of the first column) at the same time, and then their respective \textit{constant} factors (the second column) at the same time, and so on. Consequently, storing all the \textit{dummy flags} (the first column) together in memory means that they can all be retrieved very efficiently, with many being moved in one transaction (as opposed to individually locating and sending each row’s dummy flag one by one). For this reason, the two-dimensional data of the node-type subterms (and likewise for the 2D arrays of other subterm types) is better stored in column-major order, for a far more efficient (`\textit{coalesced}') memory access pattern.
\end{proof}

\section{Parallelised summation algorithm}
\label{app:para-sum-alg}

Given an array of $n$ numbers, the best method for summing (or alternatively multiplying) all its elements - via sequential (i.e. single-core CPU) computation - is to simply iterate linearly through the whole array while maintaining a total tally. This gives a runtime complexity of $O(n)$.

Alternatively, by instead processing the data on the GPU, one may compute such a calculation while taking advantage of parallelism. Consider, for example, the following 10-element array:

\begin{center}
$[\;0,\;1,\;2,\;3,\;4,\;5,\;6,\;7,\;8,\;9\;]$
\end{center}

One can allocate a single GPU thread for every adjacent pair of elements. For the general case of an $n$-element array, this means allocating $n/2$ threads, which in this example case means 5 threads, divided as follows:

\begin{center}
$[ \;\; \textbf{0},\;1, \;\;|\;\; \textbf{2},\;3, \;\;|\;\; \textbf{4},\;5, \;\;|\;\; \textbf{6},\;7, \;\;|\;\; \textbf{8},\;9 \;\;]$
\end{center}

Each thread, running in parallel, may then sum the two elements under its consideration and overwrite its left element with the result, like so:

\begin{center}
$[ \;\; \textbf{1},\;\textcolor{lightgray}{1}, \;\;|\;\; \textbf{5},\;\textcolor{lightgray}{3}, \;\;|\;\; \textbf{9},\;\textcolor{lightgray}{5}, \;\;|\;\; \textbf{13},\;\textcolor{lightgray}{7}, \;\;|\;\; \textbf{17},\;\textcolor{lightgray}{9} \;\;]$
\end{center}

The right element of each pair can hereafter be ignored and so is written in grey here. The process can then repeat - this time with $n/4$ threads being allocated to consider, respectively, adjacent groups of 4 elements, like so:

\begin{center}
$[ \;\; \textbf{1},\;\textcolor{lightgray}{1},\;5,\;\textcolor{lightgray}{3}, \;\;|\;\; \textbf{9},\;\textcolor{lightgray}{5},\;13,\;\textcolor{lightgray}{7}, \;\;|\;\; \textbf{17},\;\textcolor{lightgray}{9} \;\;]$
\end{center}

Now, for each thread, the leftmost element will be incremented by the value contained 2 elements to its right (unless that index exceeds the length of the array):

\begin{center}
$[ \;\; \textbf{6},\;\textcolor{lightgray}{1},\;\textcolor{lightgray}{5},\;\textcolor{lightgray}{3}, \;\;|\;\; \textbf{22},\;\textcolor{lightgray}{5},\;\textcolor{lightgray}{13},\;\textcolor{lightgray}{7}, \;\;|\;\; \textbf{17},\;\textcolor{lightgray}{9} \;\;]$
\end{center}

This process then repeats iterative, halving the number of threads needed each time and doubling, for each thread, the number of elements considered and the gap between its relevant elements. This procedure follows as such:

\begin{center}
$[ \;\; \textbf{6},\;\textcolor{lightgray}{1},\;\textcolor{lightgray}{5},\;\textcolor{lightgray}{3},\; 22,\;\textcolor{lightgray}{5},\;\textcolor{lightgray}{13},\;\textcolor{lightgray}{7}, \;\;|\;\; \textbf{17},\;\textcolor{lightgray}{9} \;\;]$
\end{center}

\begin{center}
$[ \;\; \textbf{28},\;\textcolor{lightgray}{1},\;\textcolor{lightgray}{5},\;\textcolor{lightgray}{3},\; \textcolor{lightgray}{22},\;\textcolor{lightgray}{5},\;\textcolor{lightgray}{13},\;\textcolor{lightgray}{7}, \;\;|\;\; \textbf{17},\;\textcolor{lightgray}{9} \;\;]$
\end{center}

\begin{center}
$[ \;\; \textbf{28},\;\textcolor{lightgray}{1},\;\textcolor{lightgray}{5},\;\textcolor{lightgray}{3},\; \textcolor{lightgray}{22},\;\textcolor{lightgray}{5},\;\textcolor{lightgray}{13},\;\textcolor{lightgray}{7},\; 17,\;\textcolor{lightgray}{9} \;\;]$
\end{center}

\begin{center}
$[ \;\; \textbf{45},\;\textcolor{lightgray}{1},\;\textcolor{lightgray}{5},\;\textcolor{lightgray}{3},\; \textcolor{lightgray}{22},\;\textcolor{lightgray}{5},\;\textcolor{lightgray}{13},\;\textcolor{lightgray}{7},\; \textcolor{lightgray}{17},\;\textcolor{lightgray}{9} \;\;]$
\end{center}

Once the segment size exceeds the length of the array, the procedure is complete and the final result (in this case $45$) - being the sum of all the elements in the original array - is stored in the first element.

Theoretically, this method computes in as few as $log_{2}n$ iterations, rather than $n$. Realistically, however, as the number of threads on a GPU is finite, for sufficiently large arrays the theoretically parallelised steps won't all occur simultaneously but rather in locally parallelised batches. Nevertheless, even as $n\rightarrow\infty$ this method requires significantly fewer iterations than the na\"ive approach. And lastly, note that the method also works likewise for multiplication as well as summation.

The pseudocode for this method is provided in algorithm \ref{alg:sum}, where the \textit{PSum} procedure computes in parallel on $n$ threads - with each given a successive thread index provided by \textit{GetThreadIndex()}.

This algorithm (or variations thereof) is used as described in the paper, for both multiplying the subterms within each term as well as summing all the terms in the overall scalar. The only major difference to note is that instead of the simple number summing calculation on line \ref{line:add}, the complex numbers are summed (and later multiplied) according to the relations described in Lemmas \ref{lemma:sumComplex} and \ref{lemma:multComplex}.

\begin{lemma}
    \label{lemma:sumComplex}
    Two complex numbers, $\psi,\phi\in\mathbb{C}$, each expressed via four simple real values, $A_x,B_x,C_x,D_x\in\mathbb{R}\;\forall x$, in the form:
    \begin{equation}
        \begin{aligned}
            \psi&=A_\psi+B_\psi\sqrt{2}+i\left(C_\psi+D_\psi\sqrt{2}\right)\\
            \phi&=A_\phi+B_\phi\sqrt{2}+i\left(C_\phi+D_\phi\sqrt{2}\right)
        \end{aligned}
    \end{equation}
    may be summed together, $\vartheta=\psi+\phi$, using only simple real arithmetic:
    \begin{equation}
        \vartheta=A_\vartheta+B_\vartheta\sqrt{2}+i\left(C_\vartheta+D_\vartheta\sqrt{2}\right)
    \end{equation}
    where here:
    \begin{equation}
        A_{\vartheta} = A_\psi + A_\phi
    \end{equation}
    \begin{equation*}
        B_{\vartheta} = B_\psi + B_\phi
    \end{equation*}
    \begin{equation*}
        C_{\vartheta} = C_\psi + C_\phi
    \end{equation*}
    \begin{equation*}
        D_{\vartheta} = D_\psi + D_\phi
    \end{equation*}
\end{lemma}

\begin{lemma}
    \label{lemma:multComplex}
    Two complex numbers, $\psi,\phi\in\mathbb{C}$, each expressed via four simple real values, $A_x,B_x,C_x,D_x\in\mathbb{R}\;\forall x$, in the form:
    \begin{equation}
        \begin{aligned}
            \psi&=A_\psi+B_\psi\sqrt{2}+i\left(C_\psi+D_\psi\sqrt{2}\right)\\
            \phi&=A_\phi+B_\phi\sqrt{2}+i\left(C_\phi+D_\phi\sqrt{2}\right)
        \end{aligned}
    \end{equation}
    may be multiplied together, $\vartheta=\psi\times\phi$, using only simple real arithmetic:
    \begin{equation}
        \vartheta=A_\vartheta+B_\vartheta\sqrt{2}+i\left(C_\vartheta+D_\vartheta\sqrt{2}\right)
    \end{equation}
    where:
    \begin{equation}
        \begin{aligned}
            A_\vartheta&=A_\psi A_\phi+2B_\psi B_\phi -C_\psi C_\phi -2D_\psi D_\phi \\
            B_\vartheta&=A_\psi B_\phi + B_\psi A_\phi -C_\psi D_\phi - D_\psi C_\phi \\
            C_\vartheta&=A_\psi C_\phi + 2B_\psi D_\phi +C_\psi A_\phi +2D_\psi B_\phi \\
            D_\vartheta&=A_\psi D_\phi + B_\psi C_\phi + C_\psi B_\phi +D_\psi A_\phi
        \end{aligned}
    \end{equation}
\end{lemma}

\begin{algorithm}
\caption{The GPU-parallelised summation algorithm}\label{alg:sum}
\begin{algorithmic}[1]

\State Initialise and populate $ARR[N_{ELEMS}]$ \Comment{The array whose elements are to be summed}

\State

\Procedure{PSum}{$split,gap$}{$<<n>>$} \Comment{This is a GPU kernel, executing on $n$ threads}
    \State $i \gets $ \Call{GetThreadIndx}{$ $} \Comment{This function returns the unique thread index}
    \State $elem \gets i \times split$
    \If{$elem + gap < N_{ELEMS} - 1$}
        \State $ARR[elem] \gets ARR[elem] + ARR[elem+gap]$ \label{line:add}
    \EndIf
\EndProcedure

\State

\State $split \gets 2$
\State $gap \gets 1$
\While{$gap < N_{ELEMS}$}
    \State $n_{chunks} \gets \lceil N_{ELEMS}/split \rceil$
    \State \Call{PSum}{$split, gap$}{$<<n_{chunks}>>$}
    \State $split \gets split \times 2$
    \State $gap \gets gap \times 2$
\EndWhile

\end{algorithmic}
\end{algorithm}

\section{Repeated weak simulation}
\label{app:repeated-weak}

In addition to the use-case of computing marginal probabilities via the summing method (left-hand side of Figure \ref{fig:classicsim}), this method may also be applied to repeated strong simulation, and indeed also repeated \textit{weak} simulation. Its relevance to the former is perhaps obvious, though to the latter may need further explanation.

In this situation, each ($1\leq k\leq n$) marginal probability computation of an $n$-qubit circuit, $U$, can be computed parametrically (see right-hand side of Figure \ref{fig:classicsim}). With this approach, at each iteration (i.e. in incrementing $k$) of computing the next marginal probability, one may deduce the next bit in the final bitstring, for any number $N$ of independent bitstrings. Hence, one will ultimately produce $N$ such bitstrings denoting $N$ independent samples of weak simulation, while only ever reducing $n$ doubled ZX-diagram, given $n$ qubits.

This is summarised in algorithm \ref{alg:paramzx:param-weak-sim}, where:
\begin{equation}
    P(\mathfrak{a}_1\cdots\mathfrak{a}_k) \overset{r}{\leftarrow} \bra{0\cdots0}{U^\dagger(\ket{\mathfrak{a}_1\cdots\mathfrak{a}_k}\bra{\mathfrak{a}_1\cdots\mathfrak{a}_k}\otimes I_{n-k})U}\ket{0\cdots0}
\end{equation}
denotes the reduction of the doubled marginal probability ZX-diagram (right-hand side of Figure \ref{fig:classicsim}) to a parameterised scalar expression, $P(\mathfrak{a}_1\cdots\mathfrak{a}_k)$, denoting its outcome probability. $P(B)$ then denotes a GPU-based evaluation of this scalar expression for the bitstring $B=a_1\cdots a_k$, where $a_i\in\mathbb{B}\;\forall i$. Meanwhile, $X||Y$ denotes the concatenation of $X$ and $Y$, such that $[a,b,c]||d=[a,b,c,d]$, and $\texttt{Rand}()$ returns a random floating point number in the range $[0,1)$.

\begin{algorithm}
    \caption{Algorithm for efficient repeated weak simulation}
    \label{alg:paramzx:param-weak-sim}
    \begin{algorithmic}[1]

    \State \textbf{Input:} A quantum circuit $U$ with $n_{qubits}$ qubits, and $N_{evals}$
    \State
    \State $B_i=[\;]\;\;\forall i\in\{1,\ldots,n_{qubits}\}$ \Comment{Initialise empty lists}
    
    \For{$k=1$ to $n_{qubits}$}
        \State $P(\mathfrak{a}_1\cdots\mathfrak{a}_k) \overset{r}{\leftarrow} \bra{0\cdots0}{U^\dagger(\ket{\mathfrak{a}_1\cdots\mathfrak{a}_k}\bra{\mathfrak{a}_1\cdots\mathfrak{a}_k}\otimes I_{n-k})U}\ket{0\cdots0}$
        
        \For{$i=1$ to $N_{evals}$}
            \State $P_{B_i||0}=P(B_i||0)$ \Comment{GPU-based evaluation}
            \State $b \gets 0\;\;\textbf{if}\;\;\texttt{Rand}() < P_{B_i||0}\;\;\textbf{else}\;1$
            \State $B_i=B_i||b$
        \EndFor
        
    \EndFor

    \State
    \State \textbf{Output:} $B_i\;\forall i$
    
    \end{algorithmic}
\end{algorithm}

There are further optimisations that may be made to this algorithm, with some particular redundancy among the first few iterations of $k$, where identical evaluations are inevitable. But such considerations are perhaps superfluous.

\section{Sigmoid curves}
\label{app:sigmoid}

The relationship between the speedup factor, $S_N$, and the number of evaluations, $N$, can be visualised as a sigmoid curve of the form:
\begin{equation}
    S_N\left(N\right)=S_{\infty}\cdot\frac{N}{N_{inflec}+N}
\end{equation}
as in Figure \ref{fig:sigmoid}, where $N_{inflec}$ is the inflection point: that is the value of $N$ at which $S_N=S_\infty/2$. In fact, $S_\infty$ may be calculated as the ratio of the evaluation times of the two methods, while $N_{inflec}$ may be calculated as the ratio of the initialisation time against the evaluation time of the new method.

\begin{figure}[h!]
    \centering
    \includegraphics[scale=0.6]{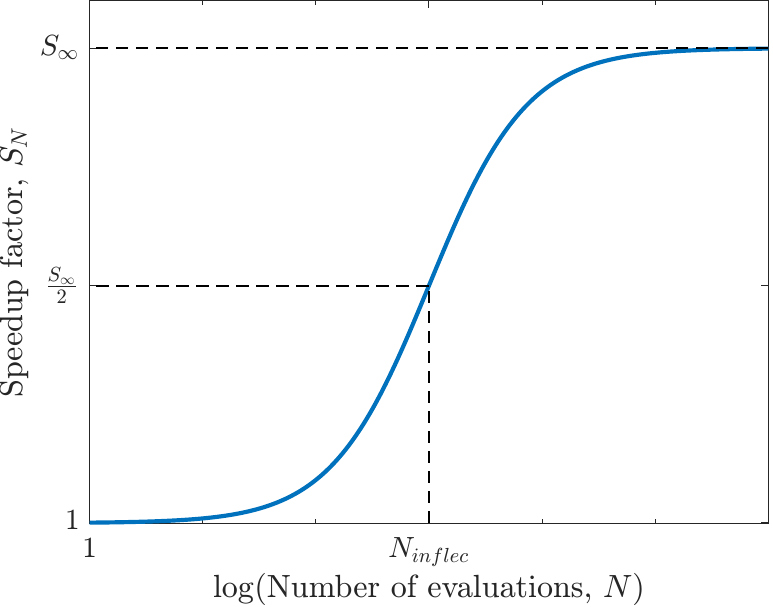}
    \caption{The relationship between the speedup factor, $S_N$, and the number of evaluations, $N$, for any particular ZX-diagram may be visualised as a sigmoid curve. Note the logarithmic $x$-axis.}
    \label{fig:sigmoid}
\end{figure}

\end{document}